%% file: main.tex
\documentclass[manuscript,screen]{acmart}
\AtBeginDocument{%
  }

\setcopyright{cc}
\setcctype{by}
\acmJournal{FAC}
\acmYear{2025} \acmVolume{1} \acmNumber{1} \acmArticle{1} \acmMonth{1} \acmPrice{}\acmDOI{10.1145/3749848}

\input{packages}
\input{macros}
\AtEndPreamble{%
    \theoremstyle{acmdefinition}%
    \newtheorem{remark}[theorem]{Remark}
}


\begin{document}

\title{Efficient Runtime Verification of Real-Time Systems under Parametric Communication Delays}

\author{Martin Fränzle}
\email{martin.fraenzle@uol.de}
\orcid{0000-0002-9138-8340}
\affiliation{%
  \institution{Carl von Ossietzky Universität Oldenburg}
  \city{Oldenburg}
  \country{Germany}
}
\author{Thomas M.\ Grosen}
\email{tmgr@cs.aau.dk}
\orcid{0009-0008-7719-6482}
\author{Kim G.\ Larsen}
\email{kgl@cs.aau.dk}
\orcid{0000-0002-5953-3384}
\author{Martin Zimmermann}
\email{mzi@cs.aau.dk}
\orcid{0000-0002-8038-2453}
\affiliation{%
  \institution{Aalborg University}
  \city{Aalborg}
  \country{Denmark}
}


\begin{abstract}
Timed Büchi automata provide a very expressive formalism for expressing requirements of real-time systems.  Online monitoring and active testing of embedded real-time systems can then be achieved by symbolic execution of such  automata on the trace observed from the system. However, this direct construction is only faithful if the observation of the trace is immediate in the sense that the monitor (or test harness, respectively) can assign exact timestamps to the actions it observes. This is rarely true in practice due to the substantial and fluctuating parametric delays introduced by the circuitry connecting the observed system to its monitoring or testing device. 

We present purely zone-based online monitoring and testing algorithms, which handle such parametric delays exactly without recurrence to costly verification procedures for parametric timed automata. We have implemented our algorithms on top of the real-time model checking tool \textsc{Uppaal}, and report on encouraging initial results.
\end{abstract}



\received{20 February 2007}
\received[revised]{12 March 2009}
\received[accepted]{5 June 2009}

\maketitle


\section{Introduction}
\label{sec:intro}

Online monitoring and testing are two important tools to achieve functional correctness of safety-critical systems. 
They analyse the execution traces observed from a system during its runtime by determining in real-time whether the observed traces satisfy the system's specification. 
Online monitoring passively observes an (execution) trace of the system. A typical application is to ensure that an unmanned aerial vehicle (UAV) stays within its safety envelope. In testing, the system is actively subjected to stimuli, which allows to cover a wider range of traces which might not be observed by passively monitoring the system. A typical application is a UAV on a test stand that allows to control wind speed and direction.

Both continuous online monitoring and testing are therefore concerned with unbounded time horizons, unlike offline monitoring where a fixed finite trace is analysed after the execution has terminated.
Hence, specifications for online monitoring and testing are typically defined over infinite traces, with the most significant approach being temporal logics. 
As specifications often include  real-time requirements, e.g., ``every request is answered within 10 milliseconds (ms)'', we focus here on metric-time temporal logics over timed words.
More precisely, we consider Metric Interval Temporal Logic (MITL) \cite{alur1996mitl}, which offers a good balance between expressiveness and algorithmic properties. 
For example, the request-response specification above is expressed by the MITL formula~$ G_{\ge 0}(\texttt{req} \rightarrow F_{\le 10} \texttt{resp})$.

While the specifications classify infinite traces, the traces observed online and to be checked against the specification remain finite. 
Nevertheless, one can still return verdicts~\cite{bauer2006monitoring}: for example, \emph{every} infinite extension of a finite trace with some request that is not answered within 10 ms violates the request-response specification above.
Hence, violation of the specification is already witnessed by such a finite trace.
Dually, consider the specification ``system calibration is completed within 500 ms'', expressed by the formula~$F_{\le 500}\texttt{cc}$ with the proposition~$\texttt{cc}$ representing the completion of calibration. 
Every infinite extension of a finite trace on which the calibration is completed within 400 ms satisfies the specification.
Hence, satisfaction of the specification is already witnessed by such a finite trace.
However, there are also traces and specifications for which no verdict can currently be drawn, like in the situation where no calibration has been observed yet at current time of 350 ms.
As usual, we capture these three situations with the three verdicts $\top$ (satisfaction for every extension), $\bot$ (violation for every extension), and $\unknown$ (inconclusive).

Online monitoring and testing can be achieved by compiling the MITL specification into an equivalent timed Büchi automaton and then symbolically executing the automaton on the observed trace of the system \cite{bauer2006monitoring,GrosenKLZ22}.
However, this approach is correct only if the actions of the system can be observed immediately by the monitor or test harness, respectively.
In practice, there is usually a communication delay between the system and the monitor or testing harness. This delay is induced by various types of circuitry at their interfaces, like technical sensors, conversion between analog and digital signals, and communication networks forwarding signals to the monitor.
We follow the approach described in McGraw-Hill's Encyclopedia of Networking and Telecommunications~\cite{mcgraw01} where a communication delay consists of a constant part (latency) and varying part (jitter).
Here, we consider a monolithic system that is monitored or tested, i.e., there is a single communication channel from the system to the monitor respectively one from the system to the test harness and one from the test harness to the system.

Due to the delay, the system and the symbolic execution are no longer synchronized but deviate by a delay, for which only bounds, yet not exact values tend to be known.
But even then, one can still provide meaningful verdicts, see Fig.~\ref{fig:example1}: 
\begin{figure*}[tb]
    \centering
    \includegraphics[width=.85\linewidth]{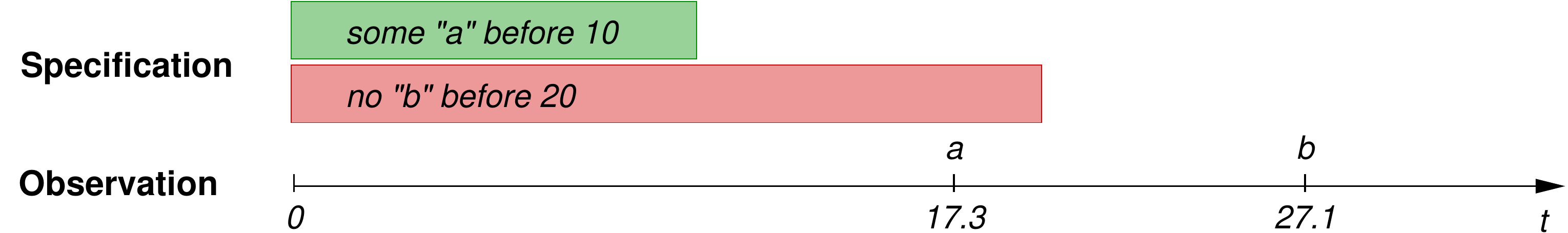}
    \caption{Monitoring under observation delay: at time $t=27.1$ we can conclusively decide that the MITL property $F_{[0,10]}a \land G_{[0,20]}\neg b$ is violated irrespective of the latency of the observation channel, provided the jitter is less than 0.2.}
    \label{fig:example1}
\end{figure*}
Consider monitoring of the specification $F_{\le 10} a \land G_{\le 20}\neg b$, which expresses that an $a$ occurs within 
10 ms and no $b$ occurs within 
20 ms, under delayed observations. The observed trace shows the first $a$ at 17.3 ms and the first $b$ at 27.1 ms. 
This observation  is only consistent with satisfaction of the constraint~$F_{\le 10} a$ if $a$'s observation delay exceeds 7.3 ms, while satisfaction of $G_{\le 20}\neg b$ requires a delay of at most 7.1 ms for $b$.
Thus, if the jitter is strictly smaller than 0.2 ms, the specification is definitely violated. Note that the verdict \myquot{violated} is true independently of the actual value of the unknown, parametric communication latency.

On the other hand, if the parametric latency is known to be in the range $[4.5,8]$ ms and the jitter is in $[0,0.3]$ ms, then we cannot give a definitive verdict: The $a$ may have occurred at 10 ms and then has been observed with 7 ms latency plus 0.3 ms jitter at 17.3 = 10 + 7 + 0.3 ms, and the $b$ may have occurred at 20.1 ms and then observed with the same latency (yet independent jitter) at 27.1 = 20.1 + 7 + 0 ms. In this case, the property would be satisfied. 
But the $a$ may also have occurred at 10.3 ms, violating the property, and still be observed with the same latency at 17.3 = 10.3 + 7 + 0 ms.
From the observations, we can nevertheless derive bounds on the parametric latency, as the property definitely is violated irrespective of the actual (unknown) value of the jitter whenever the actual latency is smaller than 7 ms or larger than 7.1 ms. It however cannot be guaranteed to be satisfied when the latency is in the range of $[7,7.1]$ ms, as satisfaction then depends on the exact value of the jitter, which is not detectable.
Thus, one can determine information beyond the verdicts~$\top$, $\bot$, and $\unknown$ in terms of bounds on the delay that imply definitive verdicts. 

\subsection{Our Contribution}

Based on previous work by Grosen et al.~\cite{GrosenKLZ22} on online monitoring of MITL specifications without delay via timed Büchi automata, we present a symbolic MITL monitoring algorithm and a symbolic MITL testing algorithm that provide exact verdicts under unknown delay consisting of parametric (i.e.\ unknown within bounds) latency and jitter. 
While an unknown delay is a timing parameter, our constructions avoid the semidecidability~\cite{DBLP:conf/stoc/AlurHV93} of analysis for parameterized timed automata, and instead uses only classical clock zones~\cite{DBLP:conf/ac/BengtssonY03}.

In addition, our approach has the advantage that it is even more informative than typical monitoring and testing algorithms, which only return a verdict in $\{\top,\bot,\unknown\}$. 
Recall the example specification~$F_{\le 10} a \land G_{\le 20}\neg b$ in the case where the jitter is constrained to [0,0.3] ms.
As argued above, this specification can, given this bound on the jitter, only be satisfied if $7 \le \ell \le 7.1$, where $\ell$ denotes the actual latency.
Our algorithms, for which we also provide a prototype implementation and experimental evaluation, compute such parametric constraints on the set of potential latencies under which the specification can be satisfied as well as on the set of potential latencies under which the specification can be violated. 

The implementation is built on top of the real-time model checking tool \textsc{Upp\-aal} \cite{DBLP:journals/sttt/LarsenPY97} using the difference-bounded matrix (DBM) data structure allowing for  representation of convex polytopes called zones.  Most importantly, the DBM data structure can be used for efficient implementation of various geometrical operations over zones needed for the symbolic analysis of timed automata, such as testing for emptiness, inclusion, equality, and computing projection and intersection of zones \cite{DBLP:conf/ac/BengtssonY03}. Our experiments show encouraging initial results on an industrial gear controller model from~\cite{DBLP:conf/tacas/LindahlPY98}.  

This article is an revised and extended version of an article published at IFM 2024~\cite{DBLP:conf/ifm/FranzleGLZ24}, which contains all proofs omitted in the conference version, a new section on testing under delay, and additional experiments.

\subsection{Related Work}
Our automata-based monitoring of finite traces against specifications over infinite words using the three verdicts~$\{\top,\bot,\unknown \}$ follows the seminal work of Bauer~et al.~\cite{bauer2006monitoring}, who presented monitoring algorithms for the qualitative-time linear temporal logic LTL and its metric-time variant Timed LTL.
Their algorithm for Timed LTL is based on clock regions~\cite{alur1994tba}, while we follow the approach of Grosen et al.~\cite{GrosenKLZ22} and leverage the performance advantages of clock zones~\cite{DBLP:conf/ac/BengtssonY03}, which account for roughly an order of magnitude of improvement in runtime (see, e.g.,~\cite{DBLP:conf/fct/LarsenPY95}).
Furthermore, Bauer et al.\ in \cite{bauer2006monitoring} translated Timed LTL into event-clock automata, which are less expressive than the timed Büchi automata (\tba) used both by Grosen et al.~\cite{GrosenKLZ22} and here.
More recently, the same approach has been used to monitor real-time properties under assumptions~\cite{cgltz24}.

As our algorithms work with \tba, we also support MITL specifications, as these can be compiled into \tba.
The monitoring problem for MITL under perfect observability, i.e.\ without delay, has been investigated before.
Baldor~et al.\ in \cite{baldor2013monitoring} showed how to construct a monitor for dense-time MITL formulas by constructing a tree of timed transducers.
Ho et al.\ split unbounded and bounded parts of MITL formulas for monitoring, using traditional LTL monitoring for the unbounded parts and permitting a simpler construction for the (finite-word) bounded parts~\cite{ho2014online}.
Bulychev et al.\ in \cite{DBLP:conf/rv/BulychevDLLLP12} apply a technique of rewriting a given Weighted Metric Temporal Logic (WMTL) formula during monitoring as part of performing statistical model checking. None of the above works makes use of the efficient difference-bounded matrix (DBM) data structure or extends to the setting of TBA that provides the basis of our approach. To this end, we note that as a specification formalism, TBA exceeds the expressive power of MITL, with the additional expressive power clearly being useful in certain application contexts (e.g., in the presence of counting properties). 

There is also a large body of work on monitoring with finite-word semantics.
Roşu et al.\ focussed on discrete-time finite-word MTL~\cite{rosu2005monitoring}, while Basin et al.\ proposed algorithms for monitoring real-time finite-word properties~\cite{basin2012algorithms} and elucidated the differences between different time models. Andr\'e et al. consider monitoring finite logs of parameterized timed and hybrid systems \cite{DBLP:journals/tcps/WagaAH22}.
Finally, Ulus et al.\ described algorithms for monitoring finite timed words against specifications given as timed regular expressions. Their monitoring procedures exploit unions of two-dimensional zones~\cite{ulus2014offline,ulus2016online}.

Common to the aforementioned technical developments is the assumption of perfect observability of the system being monitored, where the monitor can rely on always exact measurements of states and/or events of the system as well as of their occurrence times.
As this is an unrealistic assumption, the problem of monitoring trace properties under uncertain observation has been addressed before \cite{DonzeFM13,NickovicYamaguchi20,ViscontiBLN21,FFKK22,KallwiesLS22}, most notably based on Signal Temporal Logic (STL) and exploiting STL's quantitative semantics~\cite{MalerN04} that characterizes robustness against variation in state variables. 
These approaches are mostly orthogonal to ours, as they tend to address uncertainty in the state observed at a time instant rather than uncertainty in the timestamps associated to state observations. 
It would consequently be interesting to combine the two approaches, thus permitting both state uncertainty due to inexact measurements and time uncertainty due to inexact clocks and fluctuating communication latencies. 
It should also be noted that robust STL monitoring comes in diverse variants representing different error models and different levels of exactness in representing them. Early and efficient procedures for monitoring under state uncertainty implement the compositional real-valued robustness semantics \cite{DonzeFM13,NickovicYamaguchi20}. This semantics however underapproximates the factual robustness of the verdict against state shifts in the observed trace such that monitoring algorithms based on this compositional semantics are sound and computationally efficient, yet incomplete. 
Due to the safe approximation, they may yield inconclusive verdicts in actually determined situations. 
Complete and thus optimally informed STL monitoring under uncertainty, which guarantees a verdict whenever the property is determined, has only recently been investigated.  
Visconti et al.\ in \cite{ViscontiBLN21} developed sound and complete monitoring wrt.\ an interval model of state measurement error, where each single measurement features an independent displacement ranging over a bounded interval. Finkbeiner et al.\ in \cite{FFKK22,FFKK24} address a refined model distinguishing between a constant, yet unknown up to bounds, offset and a time-varying, interval-bounded noise, as suggested by the pertinent ISO norm~5725 on measurement accuracy (there called ``trueness'' and ``precision'' of a measurement). 
We here and previously in \cite{DBLP:conf/ifm/FranzleGLZ24} adopt the latter, more refined model of measurement error and transfer it into the time domain, thus implementing sound and complete monitoring for the case when timestamps are affected by a parametric (unknown, yet constant) observation latency plus a fluctuating jitter that differs between observations.
Closest to our approach is \cite{KohlH23}, which addresses a more confined model of observation delay comprising a fixed  known (non-parametric) latency plus a varying jitter. It also covers clock drift, which is an additional source of (relative) jitter that we have excluded to simplify the exposition.

Closely related to (passive) system monitoring is (active) testing, where pre-defined stimuli are supplied to the system under test (SUT) and the SUT's response is monitored in order to acquire a test verdict.
In this journal article, we expand on our previous work \cite{DBLP:conf/ifm/FranzleGLZ24} to also cover active testing, where the stimulus sequence, its timing, and the expected responses are specified together by \tba.
The work thus belongs to the large set of formal approaches to test-case generation and evaluation, which is a long-standing and  lively topic with enormous economic impact especially in the hardware domain, but also in the automotive and railway industries \cite{TeigeESSBHB21,HuangKP24}. 
Testing under timing uncertainties has been investigated from the perspective of distributed systems that exhibit asynchronous and complexly interleaved behavior due to variance in component timing \cite{DavidLLMN12,DammEA19,GraicsMMM25}. But asynchrony between the test harness and the SUT due to such timing variance in the communication infrastructure connecting the test harness and SUT has not hitherto been covered by formal approaches to testing, despite technical control of their latency and jitter having always been concerns in the design of hardware-in-the-loop testing facilities due to their fundamental impact on the feasibility of test conduct and evaluation \cite{Osorio2023}.

\section{Preliminaries}
\label{sec:prelim}

The set of natural numbers (excluding zero) is $\nats$, we define~$\nats_0 = \nats \cup \{0\}$, the set of rational numbers is $\rats$, the set of non-negative rational numbers  is $\nnrats$. 
the set of real numbers is $\reals$, and the set of non-negative real numbers  is $\nnreals$. 
The powerset of a set $S$ is denoted by $\pow{S}$.

\subsection{Timed Words}
A timed word over a finite alphabet~$\Sigma$ is a pair $\rho = (\sigma, \tau)$ where $\sigma$ is a word over $\Sigma$ and  $\tau$ is a sequence of non-decreasing non-negative real numbers of the same length as $\sigma$.
Timed words may be finite or infinite; in the latter case, we require $\limsup \tau = \infty$, i.e., time diverges. The set of finite timed words is denoted by $\TSigma^*$ and the set of infinite timed words by $\TSigma^\omega$.
We also represent a timed word as a sequence of pairs $(\sigma_1,\tau_1), (\sigma_2,\tau_2), \ldots$.  If $\rho=(\sigma_1,\tau_1), (\sigma_2,\tau_2), \ldots, (\sigma_n,\tau_n)$ is a finite timed word, we denote by $\tau(\rho)$ the total time duration of $\rho$, i.e., $\tau_n$ (with the convention that the duration of the empty word is \(0\)).

If $\rho_1=(\sigma^1_1,\tau^1_1),\ldots,(\sigma^1_n,\tau^1_n)$ is a finite timed word, $\rho_2=(\sigma^2_1,\tau^2_1),(\sigma^2_2,\tau^2_2),\ldots$ is a finite or infinite timed word, and $t \in \nnrats$ then the timed word concatenation $\rho_1 \cdot_t \rho_2$ is defined if and only if $\tau(\rho_1) \le t$. Then, $\rho_1 \cdot_t \rho_2 = (\sigma_1,\tau_1),(\sigma_2,\tau_2),\ldots$ such that 
\[\sigma_i = \begin{cases}
    \sigma^1_i & \text{if and only if } i\le n\\
    \sigma^2_{i-n} & \text{else}
\end{cases}\quad\text{and}\quad
\tau_i = \begin{cases}
    \tau^1_i & \text{if and only if } i\le n\\
    \tau^2_{i-n}+t & \text{else}.
\end{cases}\]

In the following, we often need to \myquot{shift} a timed word in the sense that we add or subtract a $t \in\nnreals$ to each time point of $\rho$.
In the latter case, we need to be careful to ensure that the time points in the shifted word are still nonnegative.
For the sake of readability, let us introduce some notation for these operations. 
Given a (finite or infinite) timed word~$\rho = (\sigma_1, \tau_1),(\sigma_2, \tau_2),\ldots$ and such a $t \in \nnreals$,
\begin{itemize}
    \item let $\rho+t$ denote the timed word~$(\sigma_1, \tau_1+t),(\sigma_2, \tau_2+t),\ldots$, and
    \item if $t \le \tau_1$, let $\rho-t$ denote the timed word~$(\sigma_1, \tau_1-t),(\sigma_2, \tau_2-t),\ldots$. This is well-defined, as we require that $\tau_1$ (and therefore each $\tau_i$) is at least $t$, so we never obtain negative time points in $\rho-t$.
\end{itemize}

The following properties follow directly from the definition of timed concatenation and will be applied in the proofs below.

\begin{remark}
\label{remark_concat}
Let $\rho = (\sigma_1, \tau_1),\ldots,(\sigma_n, \tau_n) \in \TSigma^*$, let $\mu = (\sigma_1', \tau_1'),(\sigma_2', \tau_2'),\ldots  \in \TSigma^\omega$, and let $\tau(\rho) \le t$.
\begin{enumerate}
    \item \label{remark_concat_case1}
    Let $t' \in \nnreals$ be such that $t -t' \ge \tau(\rho)$. Then
    \[
    \rho \cdot_t \mu = \rho \cdot_{t-t'} (\mu+t').
    \]
    
    \item \label{remark_concat_case2}
    Let $0 \le t' \le t$, let $n' \in \{0,1, \ldots, n\}$ be such that $\tau_{n'} \le t' \le \tau_{n'+1}$ (where we use $\tau_0 = -\infty$ to allow $n'=0$ and $\tau_{n+1} = \infty$ to allow $n' = n$) and define $\rho_1 = (\sigma_1, \tau_1),\ldots,(\sigma_{n'}, \tau_{n'})$ as well as $\rho_2 = (\sigma_{n'+1}, \tau_{n'+1}),\ldots,(\sigma_n, \tau_n)$. Then
    \[
    \rho \cdot_t \mu = \rho_1 \cdot_{t'} ( (\rho_2-t') \cdot_{t-t'}\mu ).
    \]
    
    \item \label{remark_concat_case3}
    Let $t' \ge 0$ be such that $\tau(\rho) \le t-t'$, let $n' \ge 0$ be such that $\tau_{n'}' \le t' \le \tau_{n'+1}'$ (where we use $\tau_0 = -\infty$ to allow $n'=0$), and define $\rho' = (\sigma_1', \tau_1'),\ldots,(\sigma_{n'}', \tau_{n'}')$ as well as $\mu' = (\sigma_{n'+1}', \tau_{n'+1}'),(\sigma_{n'+2}', \tau_{n'+2}'),\ldots$. Then
    \[
    \rho \cdot_{t-t'} \mu = (\rho \cdot_{t - t'} \rho') \cdot_t (\mu' - t').
    \]
    Note that $n'$ is well-defined due to time-divergence of $\mu$.
\end{enumerate}
\end{remark}

\subsection{Timed Automata}

A timed Büchi automaton~(\tba)~$\aut = (Q, Q_0, \Sigma, C, \Delta, \mathcal{F})$ consists of a finite alphabet~$\Sigma$, a finite set~$Q$ of locations, a set~$Q_0 \subseteq Q$ of initial locations, a finite set~$C$ of clocks, a finite set~$\Delta \subseteq Q \times Q \times \Sigma \times \pow{C} \times G(C)$ of transitions with $G(C)$ being the set of clock constraints over $C$, and a set~$\mathcal{F} \subseteq Q$ of accepting locations.
A transition~$(q,q',a,\lambda, g)$ is an edge from $q$ to $q'$ on input symbol $a$, where $\lambda$ is the set of clocks to reset and $g$ is a clock constraint over $C$.
A clock constraint is a conjunction of atomic constraints of the form $x \sim n$, where $x$ is a clock, $n \in \nats_0$, and $\sim\ \in \{<, \leq, =, \geq, >\}$.
A state of $\aut$ is a pair $(q,v)$ where $q$ is a location in $Q$ and $v \colon C \rightarrow \nnreals$ is a valuation mapping clocks to their values.
For any $d \in \nnreals$, $v+d$ is the valuation~$x \mapsto v(x) + d$.

A run of $\aut$ from a state $(q_0, v_0)$  over a timed word $(\sigma_1, \tau_1)  (\sigma_2, \tau_2)\cdots$ is a sequence of steps
$(q_0, v_0) \transition{1} (q_1,v_1) \transition{2} (q_2,v_2) \transition{3} \cdots $
where for all $i \geq 1$ there is a transition $(q_{i-1},q_{i},\sigma_{i},\lambda_i,g_i)$ such that $v_{i}(x) = 0$ for all $x$ in $\lambda_i$ and $v_{i}(x) = v_{i-1}(x) + (\tau_i - \tau_{i-1})$ otherwise, and $g_i$ is satisfied by the valuation $v_{i-1}+(\tau_{i} - \tau_{i-1})$.
Here, we use $\tau_0 = 0$.
Given a run $r$, we denote the set of locations visited infinitely many times by $r$ as $\infset{r}$.
A run $r$ of $\aut$ is accepting if $\infset{r} \cap \mathcal{F} \neq \emptyset$.
The language of $\aut$ from a starting state $(q,v)$, denoted $\lang{\aut,(q,v)}$, is the set of all infinite timed words with an accepting run in $\mathcal{A}$ starting from $(q,v)$. 
We define the language of $\aut$, written $\lang{\aut}$, to be $\bigcup_q \lang{\aut,(q,v_0)}$, where $q$ ranges over $Q_0$ and where $v_0(x) = 0$ for all $x \in C$.

\subsection{Logic}
We use Metric Interval Temporal Logic (\mitl) to express properties to be monitored; these are subsequently translated into equivalent \tba which we use in our monitoring algorithm.
The syntax of \mitl formulas over a finite alphabet~$\Sigma$ is defined as
\[
\varphi ::= p \mid  \neg \varphi \mid  \varphi \vee \varphi \mid  X_I \varphi \mid  \varphi\ U_{I} \varphi\
\]
where $p \in \Sigma$ and $I$ ranges over non-singular intervals over $\nnreals$ with endpoints in ${\nats_0 \cup \{\infty\}}$.
We  write $\shortsim\,n$ for $I=\{d\in\nnreals\mid d\shortsim n\}$ for $\shortsim \in \{<,\leq,\geq,>\}$ and~$n \in \nats$.
 We also define the standard syntactic sugar:
$\true = p \vee \neg p$,
$\varphi \wedge \psi = \neg (\neg \varphi \vee \neg \psi)$,
$F_{I} \varphi = \true\ U_{I} \varphi$, 
and $G_{I} \varphi = \neg F_{I} \neg \varphi$.

The satisfaction relation~$\rho, i \models \varphi$ is defined for infinite timed words~$\rho = (\sigma_1,\tau_1),(\sigma_2,\tau_2),\ldots$, positions~$i \geq 1$, and an \mitl formula~$\varphi$:
\begin{itemize}
    \item $\rho,i \models p $  if $p = \sigma_i$.
    \item $\rho,i \models \neg \varphi$ if $\rho,i \not\models \varphi$.
    \item $\rho,i \models \varphi \vee \psi$  if $\rho,i \models \varphi$ or $\rho,i \models \psi$.
    \item $\rho,i \models X_I \varphi$ if $\rho,(i+1) \models \varphi$ and $\tau_{i+1} - \tau_i \in I$.
    \item $\rho,i \models \varphi\ U_{I} \psi$ if there exists $k \geq i$ s.t.\ $\rho,k \models \psi$, $\tau_k - \tau_{i} \in I$, and $\rho,j \models \varphi$ for all $i \leq j < k$. 
\end{itemize}
  We write $\rho \models \varphi$ whenever $\rho, 1 \models \varphi$, and say that $\rho$ satisfies $\varphi$.
The language $\lang{\varphi}$ of an \mitl formula~$\varphi$ is the set of all $\rho \in \TSigma^\omega$ such that~$\rho\models\varphi$.

\begin{theorem}[\cite{alur1996mitl,brihaye2017mightyl}]
\label{thm:mtltba}
For each \mitl formula $\varphi$ there exists a \tba $\aut$ with $\lang{\varphi} = \lang{\aut}$.
\end{theorem}

Fig.~\ref{fig:aut-example} illustrates Theorem~\ref{thm:mtltba} by providing \tba\ for the formula~$F_{[0,10]}a \land G_{[0,20]}\neg b$ from the introduction and its negation.
\begin{figure}[h]
    \centering
    \input{figures/automaton-example}
    \caption{An automaton for the language of the property~$\varphi = F_{[0,10]}a \land G_{[0,20]}\neg b$ and its negation: If location~$\varphi$ is accepting then it accepts $L(\varphi)$, if location~$\neg\varphi$ is accepting then it accepts $L(\neg \varphi)$.} 
    \label{fig:aut-example}
\end{figure}

\section{Monitoring under Delayed Observation}
\label{sec:MonitoringDelay}

According to McGraw-Hill's Encyclopedia of Networking and Telecommunications~\cite{mcgraw01}, a communication delay consists of a constant part (latency) and varying part (jitter). We describe the delay as a pair $(\lat, \jit) \in \nnreals^2$ where $\lat$ is the constant latency for all signals and $\jit$ is the bound on the jitter. Thus, all signals from the system are delayed within $[\lat, \lat + \jit]$ before they arrive at the monitor.

In the simplest case, our obligation is to monitor violation of an MITL specification $\varphi$ by a system while observing the events through a channel $\chan$ featuring a constant, yet unknown (up to a given, but maybe trivial, lower bound $l \in \nnreals$ and upper bound $u \in \nnreals$) transportation latency $\lat \in [\ell,u]$ and a varying jitter bounded by $\jit \in \nnreals$. 
Fig.~\ref{fig:example1} shows an example of a property and an observation that conclusively violates the specification at time 27.1, even if the channel latency $\lat \in[0,\infty[$ is unknown, as long as the 
jitter is bounded by $0.2$.

Thus, we need to distinguish between observations (the timed word corresponding to the events as they are observed by the monitoring device, subject to delay) and the possible ground-truths, as they may have been emitted by the monitored system. 
We begin by formalizing the concept of observation, where the occurrence of observed events is constrained by a set~$\delays$ capturing known bounds on the delay.
Obviously, under latency~$\lat$ (and arbitrary jitter), the first observation can only be made after at least $\lat$ units of time.

\begin{definition}
A delay set~$\delays$ is a nonempty subset of $\nnreals^2$ containing pairs of latencies and jitters.   
A $\delays$-observation, i.e.\ an observation that can in principle be made under delay in $\delays$, is a finite timed word~$\rho^*=(\sigma^*_1,\tau^*_1),\ldots,(\sigma^*_m,\tau^*_m)$ with $\tau_1^* \ge \lat$ for some $(\lat,\jit) \in\delays$. 
\end{definition}

As the ground-truth occurrence times of events in the system cannot be determined exactly from their delayed copies that the monitor receives through the communication channel, we have to consider all ground-truth timed words that the particular observation is consistent with, as follows.\footnote{\label{footnote:overtake}Note that we simplify our definitions by assuming that jitter does not change the order of observations. Under the additional assumption that only a (uniformly) bounded number of events can be generated by the system in each unit of time, it is possible to take \myquot{overtaking} of events into account, by looking at all consistent permutations. However, this would lead to a severe overhead in the implementation.
}

\begin{definition}[Consistency]\label{def:Consistency}
   Let $\rho^*=(\sigma^*_1,\tau^*_1),\ldots,(\sigma^*_m,\tau^*_m)$ be a $\{(\lat,\jit)\}$-observation and let $\rho=(\sigma_1,\tau_1),\ldots,(\sigma_n,\tau_n)$ be  a finite timed  word. 
   We say that $\rho$ is \emph{consistent} with $\rho^*$ at observation time $t \in \nnreals$ under latency $\lat$ and jitter $\jit$ if and only if
    \begin{enumerate}
       \item $\tau_n \le t$ and $\tau^*_m \le t$,
       \item $n \ge m$, and $\sigma_i = \sigma^*_i$ and $\tau_i^* -\tau_i \in [\lat,\lat+\jit]$  for all $i \in \{1,\ldots,m\}$,\footnote{Note that the conference version~\cite{DBLP:conf/ifm/FranzleGLZ24} contained a typo in this condition.} and
        \item if $ n > m $ then $ \tau_{m+1} \ge t - (\lat + \jit )$.
   \end{enumerate}
   We denote the set of timed words $\rho$ that are consistent with a $\{(\lat,\jit)\}$-ob\-ser\-vation~$\rho^*$ at observation time $t$ under latency~$\lat$ and jitter $\jit$ by $\GT_{\lat,\jit}(\rho^*,t)$.
   Then, we define $\GT_\delays(\rho^*,t)=\bigcup_{(\lat,\jit)\in\delays}\GT_{\lat,\jit}(\rho^*,t)$.
\end{definition}

$\GT_\delays(\rho^*,t)$ thus collects the possible ground-truths that are consistent with the observation $\rho^*$ when the time elapsed since the system has started is $t$, and the delay $(\lat,\jit)$ is within the set $\delays$. 

\begin{example}
\label{example_gtdef}
Fig.~\ref{fig_gtdef} shows a $\{(\lat,\jit)\}$-observation and a consistent ground-truth and illustrates how the delay shifts the timestamps of the events. The length of $\rho$ is $n=9$ and the length of $\rho^*$ is $m=4$. Recall that $t$ is the time of observation.

\begin{figure}
\centering
    \begin{tikzpicture}[thick,xscale=1.25]
    \def\yunit{.1}

    \draw (0,0) -- (10,0);
    \draw (0,-2) -- (10,-2);

    \node[anchor = east] at (-.2,1) {time};
    \node[anchor = east] at (-.2,0) {$\rho^*$};    
    \node[anchor = east] at (-.2,-2) {$\rho$};

    
    \draw[thin, densely dotted, red] (0,.75) -- (0,-2);
    \node[] at (0,1) {$0$};


    \node[] at (8.5,1) {$t - (\lat + \jit)$}; 
    \draw[thin, densely dotted, red] (8.5,.75) -- (8.5,-2);

    \node[] at (10,1) {$t$};    
    \draw[thin, densely dotted, red] (10,.75) -- (10,-2);

    \draw[thin, densely dotted, red] (1.5,.75) -- (1.5,-2.75);
    \node[anchor = south,fill=white, inner sep=0] at (1.5,2.5*\yunit) {a};
    \node[] at (1.5,1) {$\tau_1^*$};
    \draw (1.5,\yunit) -- (1.5,-\yunit);

    \draw[thin, densely dotted, red] (4.25,.75) -- (4.25,-2.75);
    \node[anchor = south,fill=white, inner sep=0] at (4.25,2.5*\yunit) {a};
    \node[] at (4.25,1) {$\tau_2^*$};
    \draw (4.25,\yunit) -- (4.25,-\yunit);

    \draw[thin, densely dotted, red] (6.5,.75) -- (6.5,-2.75);
    \node[anchor = south,fill=white, inner sep=0] at (6.5,2.5*\yunit) {b};
    \node[] at (6.5,1) {$\tau_3^*$};
    \draw (6.5,\yunit) -- (6.5,-\yunit);

    \draw[thin, densely dotted, red] (7.25,.75) -- (7.25,-3.5);
    \node[anchor = south,fill=white, inner sep=0] at (7.25,2.5*\yunit) {a};
    \node[] at (7.25,1) {$\tau_4^*$};
    \draw (7.25,\yunit) -- (7.25,-\yunit);
    
    \draw[thin, densely dotted, red] (0.5,.75) -- (0.5,-2.75);
    \node[anchor = south,fill=white, inner sep=0] at (0.5,-2-5*\yunit) {a};
    \node[] at (0.5,1) {$\tau_1$};
    \draw (0.5,-2+\yunit) -- (0.5,-2-\yunit);
    
    \draw[thin, densely dotted, red] (3,.75) -- (3,-2.75);
    \node[anchor = south,fill=white, inner sep=0] at (3,-2-5*\yunit) {a};
    \node[] at (3,1) {$\tau_2$};
    \draw (3,-2+\yunit) -- (3,-2-\yunit);
    
    \draw[thin, densely dotted, red] (5.5,.75) -- (5.5,-2.75);
    \node[anchor = south,fill=white, inner sep=0] at (5.5,-2-5*\yunit) {b};
    \node[] at (5.5,1) {$\tau_3$};
    \draw (5.5,-2+\yunit) -- (5.5,-2-\yunit);
    
    \draw[thin, densely dotted, red] (5.75,1.1) -- (5.75,-3.5);
    \node[anchor = south,fill=white, inner sep=0] at (5.75,-2-5*\yunit) {a};
    \node[] at (5.75,1.25) {$\tau_4$};
    \draw (5.75,-2+\yunit) -- (5.75,-2-\yunit);

    \draw[loosely dashed,stealth-, shorten >=3pt, shorten <=3pt] (1.5,0) -- (.5,-2);
    \draw[loosely dashed,stealth-, shorten >=3pt, shorten <=3pt] (4.25,0) -- (3,-2);
    \draw[loosely dashed,stealth-, shorten >=3pt, shorten <=3pt] (6.5,0) -- (5.5,-2);
    \draw[loosely dashed,stealth-, shorten >=3pt, shorten <=3pt] (7.25,0) -- (5.75,-2);

    \draw [<->, >=stealth] (.5,-2.75) -- (2,-2.75)
    node [black,midway,yshift=-5pt] {$\lat + \jit$};

    \draw [<->, >=stealth] (3,-2.75) -- (4.5,-2.75)
    node [black,midway,yshift=-5pt] {$\lat + \jit$};

    \draw [<->, >=stealth] (5.5,-2.75) -- (7,-2.75)
    node [black,midway,yshift=-5pt] {$\lat + \jit$};

    \draw [<->, >=stealth] (5.75,-3.5) -- (7.25,-3.5)
    node [black,midway,yshift=-5pt] {$\lat + \jit$};

    \node[anchor = north] at (8.5,-2-\yunit) {a};
    \draw (8.5,-2+\yunit) -- (8.5,-2-\yunit);
    
    \node[anchor = north] at (8.75,-2-\yunit) {a};
    \draw (8.75,-2+\yunit) -- (8.75,-2-\yunit);
    
    \node[anchor = north] at (9,-2-\yunit) {b};
    \draw (9,-2+\yunit) -- (9,-2-\yunit);
    
    \node[anchor = north] at (9.5,-2-\yunit) {b};
    \draw (9.5,-2+\yunit) -- (9.5,-2-\yunit);
    
    \node[anchor = north] at (9.75,-2-\yunit) {b};
    \draw (9.75,-2+\yunit) -- (9.75,-2-\yunit);

    \end{tikzpicture}
    \caption{A $\{(\lat,\jit)\}$-observation~$\rho^*$ and a consistent ground-truth~$\rho$.}
    \label{fig_gtdef}
\end{figure}

In particular, notice the following:
\begin{itemize}
    \item No event can occur in the observation~$\rho^*$ with a timestamp smaller than $\lat$, as it takes at least $\lat$ units of time for an event to be send from the system through the communication channel to the monitor. Obviously, at the system side (i.e., in the ground-truth~$\rho$) events can happen at any timestamp, also before $\lat$ (e.g., the first $a$).

    \item The difference~$\tau_i^*-\tau_i$ for $i \le 4$ (i.e., the difference between the time the event is observed and the time the event was emitted) must be in the interval~$[\lat,\lat+\jit]$.

    \item The time elapsed between the $b$ and the last $a$ in the observation~$\rho^*$ is larger than the time elapsed between the corresponding events in the ground-truth~$\rho$. This means the jitter for the $a$ is larger than the jitter for the $b$.

    \item The last five events in $\rho$ have not yet been observed in $\rho^*$. Such events can only have timestamps in the interval~$[t-(\lat+\jit), t]$, as all earlier events must necessarily have been observed. Said differently, there cannot be any events between timestamp~$\tau_4$ (corresponding to the last observed event in $\rho^*$ with timestamp~$\tau_4^*$) and timestamp~$t-(\lat+\jit)$, as any such event would have arrived at the monitor, even under the maximal possible delay of $\lat+\jit$.

    However, there can be an arbitrary number of events in the ground-truth~$\rho$ between timestamps~$t-(\lat+\jit)$ and $t$.
\end{itemize}
\end{example}

\begin{remark}
\label{remark_gtnonempty}
Note that $\GT_\delays(\rho^*,t)$ is always nonempty, if $\rho^*$ is a $\delays$-observation and $t \ge \tau(\rho^*)$, as it contains, e.g., $\rho^* - \lat$ for a $(\lat,\jit)\in\delays$ such that $\rho^*$ is a $(\lat,\jit)$-observation.
\end{remark}

The following property about consistent words will be useful in the proofs below.
Here, we use an extension relation over observations with time points:
Let $\rho=(\sigma_1,\tau_1),\ldots,(\sigma_n,\tau_n)$ and $\rho'=(\sigma'_1,\tau'_1),\ldots,(\sigma'_{n'},\tau'_{n'})$ be two finite timed words. Also let $t, t' \in \nnreals$ with $\tau(\rho)\le t$ and $\tau(\rho')\le t'$.  
Then, we define $(\rho,t)\sqsubseteq (\rho',t')$, if $n\leq n'$,  $\sigma_i=\sigma'_i $ and $\tau_i=\tau'_i$ for all $i\leq n$, and either $n = n'$ and $t\leq t'$ or $n< n'$ and $t\leq \tau'_{n+1}$.

\begin{lemma}
\label{lemma_GT}
Let $(\rho^*_1,t_1) \sqsubseteq (\rho_2^*, t_2)$ with $\tau(\rho_1^*) \le t_1$ and $\tau(\rho_2^*) \le t_2$, let $\rho_2 = (\sigma_1, \tau_1),\ldots,(\sigma_n,\tau_n) \in \GT_{\lat,\jit}(\rho_2^*, t_2)$, let $n' \in \{0,1,\ldots,n\}$ be such that $\tau_{n'} \le t_1 \le \tau_{n'+1}$ (where we use $\tau_0 = -\infty$ to allow $n'=0$ and $\tau_{n+1} = \infty$ to allow $n' = n$), and define $\rho_1 = (\sigma_1, \tau_1),\ldots,(\sigma_{n'},\tau_{n'})$. Then $\rho_1 \in \GT_{\lat,\jit}(\rho_1^*, t_1)$.
\end{lemma}

\begin{proof}
We need to show that $\rho_1$ is consistent with $\rho_1^*$ at $t_1$ under $\lat$ and $\jit$. The first requirement of the definition of consistency follows from $\tau(\rho_1^*) \le t_1$ and the choice of $n'$ (which implies $\tau(\rho_1) \le t_1$).
The second requirement follows from the fact that $\rho_2$ is consistent with $\rho_2^*$ at $t_2$ under $\lat$ and $\jit$ and the fact that $\rho_1$ is a prefix of $\rho_2$ and $\rho_1^*$ is a prefix of $\rho_2^*$.

Finally, consider the third requirement and assume it is violated, i.e., let $\rho_1^*$ have $m$ letters and assume $\rho_1$ has at least $m+1$ letters, i.e., the $(m+1)$-th letter of $\rho_1$ has not yet been observed before time~$t_1$ in $(\rho_1^*)$.
Let $\tau_{m+1}$ denote the time point of the $(m+1)$-th letter of $\rho_1$, which is also the time-point of the $(m+1)$-th letter of $\rho_2$. 
We consider two cases.

If $\rho_1^*$ and $\rho_2^*$ have the same length, then the $(m+1)$-th letter of $\rho_1$ (which is also the $(m+1)$-th letter of $\rho_2$) has also not been observed before time~$t_2$ in $(\rho_2^*)$.
Hence, the third requirement of consistency (between $\rho_2^*$ and $\rho_2$) implies $t_{m+1} \ge t_2 - (\lat + \jit)$. Hence, we obtain the desired bound~$t_{m+1} \ge t_1 - (\lat + \jit)$ from the inequality~$t_2 \ge t_1$.

The only other option is that $\rho_1^*$ is strictly shorter than $\rho_2^*$. 
This implies that the $(m+1)$-th letter of $\rho_2$ (which is also the $(m+1)$-th letter of $\rho_1$) has been observed before time~$t_2$ in $\rho_2^*$.
Then, the second requirement of consistency (between $\rho_2^*$ and $\rho_2$) implies $\tau_{m+1}^* - \tau_{m+1} \in [\lat, \lat+\jit]$, where $\tau_{m+1}^*$ is the time-point at which the $(m+1)$-th letter of $\rho_2$ is observed in $\rho_2^*$.
Furthermore, the definition of $\sqsubseteq$ yields $\tau_{m+1}^* \ge t_1$. 
Combining these two yields the desired bound~$t_{m+1} \ge t_1 - (\lat + \jit)$.
\end{proof}

Next, we introduce monitoring under delay. A monitor obviously ought to supply a verdict if and only if that verdict applies across \emph{all possible} ground-truth timed words that the observed word explains. 
For our definition of monitor, we use the set~$\verdicts = \{\top, \unknown,\bot\}$ of verdicts, as usual.

\begin{definition}[Monitor verdicts under delay]
\label{def:delayedmonitor}
Given a language $L \subseteq \TSigma^\omega$, a set of possible observation delays $\delays$, a $\delays$-observation~$\rho^* \in \TSigma^*$, and an observation time~$t \ge \tau(\rho^*)$, the function $\evalfuncsymbol_{\delays} \colon \pow{\TSigma^\omega}\rightarrow \TSigma^*  \times \nnreals \rightarrow \verdicts$ evaluates to the verdict
\[
\evalfuncd{\delays}{L}{\rho^*,t} = \left.
  \begin{cases}
    \top & \text{if } \rho \cdot_t \mu \in L $ for all $\rho \in \GT_{\delays}(\rho^*,t)$ and all $ \mu \in \TSigma^{\omega}, \\
    \bot & \text{if } \rho \cdot_t \mu \notin L $ for all $\rho \in \GT_{\delays}(\rho^*,t)$ and all $ \mu \in \TSigma^{\omega}, \\
    \unknown & \text{otherwise}.
  \end{cases}
  \right.
\]
$\evalfuncd{\delays}{L}{\rho^*,t}$ is undefined when $t < \tau(\rho^*)$.
\end{definition}

\begin{example}
\label{example:verdicts}
    Consider the property $\varphi = F_{[0,10]}a \land G_{[0,20]}\neg b$ and observed word $\rho^*=(a, 17.3), (b, 27.1)$ shown in Fig.~\ref{fig:example1}, time point $t = 27.1$, and set of delays~$\delays = \{(\lat, 0.2) \mid \lat \in\nnreals\}$.
    As the jitter is bounded by $0.2$, in all ground-truths either $a$ occurred after time point~10, or $b$ occurred before time point~20. Thus, all extensions of all possible ground-truths satisfy $\neg\varphi$, i.e., $\evalfuncd{\delays}{L(\varphi))}{\rho^*,t} = \bot$.
\end{example}

Note that for the special case of $\delays=\{(0,0)\}$ we cover classical (i.e., delay-free) monitoring~\cite{GrosenKLZ22}.
Before we turn our attention to computing $\evalfuncsymbol$ in Sections~\ref{sec:towards} and \ref{sec:algo}, we study some properties of our definition. First, let us note that the ability to make firm verdicts increases with increased certainty of the observation channel delay.

\begin{lemma}
\label{lemma:delaycertainty}
Let $L \subseteq \TSigma^\omega$, $\rho^* \in \TSigma^*$, let $\delays \subseteq \delays'$ be delay sets, let $\rho^*$ be a $\delays$-observation, and let $t \ge \tau(\rho^*)$. 
Then, $\evalfuncd{\delays'}{L}{\rho^*,t}=\top$ implies $\evalfuncd{\delays}{L}{\rho^*,t}=\top$ and $\evalfuncd{\delays'}{L}{\rho^*,t}=\bot$ implies $\evalfuncd{\delays}{L}{\rho^*,t}=\bot$.
\end{lemma}

\begin{proof}
Note that $\delays \subseteq \delays'$ implies $\GT_\delays(\rho^*,t) \subseteq \GT_{\delays'}(\rho^*,t)$. Thus, the universal quantification over possible ground-truths~$\rho$ in the first two cases of the definition of $\evalfuncd{\delays}{L}{\rho^*,t}$ ranges over a subset of the possible ground-truths that are considered for $\evalfuncd{\delays'}{L}{\rho^*,t}$.
\end{proof}

As a refinement of the verdict function in Definition~\ref{def:delayedmonitor}, one may provide information about the delay parameters~$(\lat,\jit)$ that can explain an observation.  
Given~$L \subseteq \TSigma^\omega$, a finite timed word~$\rho^* \in \TSigma^*$, and $t \ge \tau(\rho^*)$, the set~$\Delta(L,\rho^*,t)$ of delays that are consistent with the observation $\rho^*$ at $t$ is defined as
\[ \Delta(L,\rho^*,t) = \{ (\lat,\jit) \mid 
    \exists \rho\in\GT_{\lat,\jit}(\rho^*,t)\, \exists \mu\in\TSigma^{\omega} \text{ s.t. }\rho \cdot_t \mu \in L\}. \]
We denote by $\Delta_{\delays}(L,\rho^*,t)$ the set $\Delta(L,\rho^*,t) \cap \delays$.

The following fact follows directly from the nonemptiness of ground-truths and is later useful in proofs. 

\begin{lemma}
\label{lemma_consdelayunion}
Let $\rho^*$ be a $(\lat,\jit)$-observation, $t \ge \tau(\rho^*)$, and $L \subseteq \TSigma^\omega$.
Then, $(\lat,\jit) \in \Delta(L,\rho^*,t)$ or $(\lat,\jit) \in \Delta(\overline{L},\rho^*,t)$ (note that it may be in both).
\end{lemma}

\begin{proof}
Fix some $\rho \in \GT_{\lat,\jit}(\rho^*,t)$ (which is always possible due to Remark~\ref{remark_gtnonempty}) and some $\mu \in \TSigma^\omega$.
Then, $\rho \cdot_t \mu$ is either in $L$ or in $\overline{L}$.
In the former case, we have $(\lat,\jit) \in \Delta(L,\rho^*,t)$, in the latter case, we have $(\lat,\jit) \in \Delta(\overline{L},\rho^*,t)$.
\end{proof}

In the following, we present an example of a delay that is in both sets of consistent delays.

\begin{example}
Let $L = \lang{F_{\le 10} a}$, consider the observation~$\rho^* = (b,3)$, and let $\delays = \{ (3,7) \}$.
Then, we have $(3,7) \in \Delta(L,\rho^*,t)$ (as we can extend the ground-truth~$(b,0)$ so that it is in $L$) and $(3,7) \in \Delta(\overline{L},\rho^*,t)$ (as we can extend the ground-truth~$(b,0)$ so that it is in $\overline{L}$).
\end{example}

Our next result shows that conclusive monitoring verdicts can be characterized via the sets of consistent delays, i.e., computing the sets of consistent delays generalizes monitoring under delay.

\begin{lemma}
\label{lemma:delay}
Given $L \subseteq \TSigma^\omega$, a set~$\delays$ of delays, a $\delays$-observation~$\rho^* \in \TSigma^*$, and $t \ge \tau(\rho^*)$, we have
    \begin{enumerate}
       \item  $\Delta_{\delays}(L,\rho^*,t)=\emptyset$ if and only if $\evalfuncd{\delays}{L}{\rho^*,t}=\bot$, and 
       \item  $\Delta_{\delays}(\overline{L},\rho^*,t)=\emptyset$ if and only if $\evalfuncd{\delays}{L}{\rho^*,t}=\top$.
    \end{enumerate}
\end{lemma}

\begin{proof}
 We have
\begin{align*}
  {}&{}  \Delta_{\delays}(L,\rho^*,t)=\emptyset\\
 \Leftrightarrow {}&{} \rho\cdot_t\mu \in\overline{L} \text{ for all $(\lat,\jit) \in \delays$, all $\rho \in \GT_{\lat,\jit}(\rho^*,t)$, and all $\mu \in \TSigma^\omega$}\\
 \Leftrightarrow {}&{} \rho\cdot_t\mu \in\overline{L} \text{ for all $\rho \in \GT_{\delays}(\rho^*,t)$ and all $\mu \in \TSigma^\omega$}\\
 \Leftrightarrow {}&{} \evalfuncd{\delays}{L}{\rho^*,t}=\bot.
\end{align*}

The second claim is obtained by a dual argument (swapping $\bot$ with $\top$ and $L$ with $\overline{L}$).
\end{proof}

But even in the case when both delay-sets are nonempty (i.e., the verdict is $\unknown$), we can still provide useful information in terms of the sets~$\Delta(L, \rho^*, t)$ and $\Delta(\overline{L}, \rho^*, t)$ of consistent delays. 
In particular, the set of consistent delays is non-increasing during observations: By extending the observations, we (potentially) reduce the set of consistent delays.

\begin{lemma}
\label{lemma:delaymono}
Let $(\rho_1^*,t_1)\sqsubseteq (\rho_2^*,t_2)$ for finite timed words~$\rho_1^*$ and $\rho_2^*$ with $t_1 \ge \tau(\rho_1^*)$ and $t_2 \ge \tau(\rho_2^*)$ and $t_2 \ge t_1$. Then, $\Delta(L,\rho_1^*,t_1)\supseteq\Delta(L,\rho_2^*,t_2) $. 
\end{lemma}

\begin{proof}
Let $(\lat,\jit) \in \Delta(L,\rho_2^*,t_2)$, i.e., there exists $\rho_2 \in \GT_{\lat,\jit}(\rho_2^*,t_2)$ and a $\mu_2 \in \TSigma^\omega$ such that $\rho_2\cdot_{t_2}\mu_2 \in L$. 

Let $\rho_2 = (\sigma_1, \tau_1), \ldots,(\sigma_n, \tau_n)$ and let $n'$ be maximal with $\tau_{n'} \le t_1$ (where we use $\tau_0 = -\infty$ to allow $n'=0$ and $\tau_{n+1} = \infty$ to allow $n' = n$). 
Then, Lemma~\ref{lemma_GT} yields
\[
\rho_1 = (\sigma_1, \tau_1), \ldots,(\sigma_{n'}, \tau_{n'}) \in \GT_{\lat,\jit}(\rho_1^*,t_1)
\]
and an application of Remark~\ref{remark_concat}.\ref{remark_concat_case2} yields
\[
\rho_1 \cdot_{t_1} \Big(\big[ ( (\sigma_{n'+1}, \tau_{n'+1}), \ldots,(\sigma_n, \tau_n ))  - t_1 \big] \cdot_{t_2-t_1} \mu\Big)  = \rho_2\cdot_{t_2}\mu_2 \in L .
\]
This implies $(\lat,\jit) \in \Delta(L,\rho_1^*,t_1)$.
\end{proof}

Another interesting point is that in some cases, no extension of the observed word will provide a definitive verdict. 
\begin{example}\label{ex:foreverinconclusive}
Consider the language $L(F_{\le 10}a)$, the observation~$\rho^* = (a,15)$, and the set~$\delays = \{(\lat, 0) \mid \lat \in [0,10]\}$ of delays.
For any given $t \ge \tau(\rho^*)$ the sets of consistent delays are $\Delta_\delays(L, \rho^*, t) = \{(\lat, 0) \mid \lat \in [5, 10]\}$ and $\Delta_\delays(\overline{L}, \rho^*, t) = \{(\lat, 0) \mid \lat \in [0, 5)\}$, i.e., both sets of consistent delays are a strict subset of $\delays$.
Further, due to Lemma~\ref{lemma:delaymono}, this will be the case, no matter what observations occur in the future, as the set of consistent delays can only shrink when further observations are made.

As the sets of consistent verdicts can only shrink, but must contain every possible~$(\lat,0) \in \delays$ (due to the fact that $\rho^*$ is a $(\lat,0)$-observation for each such $(\lat,0)$ and due to Lemma~\ref{lemma_consdelayunion}), we can conclude that neither of the sets can become empty. 
So, Lemma~\ref{lemma:delay} implies that the verdict is $\unknown$, even if additional observations occur.
\end{example}

The following lemma formalizes this: as soon as the set of consistent delays w.r.t.\ $L$ ($\overline{L}$) is no longer equal to $\delays$, then the verdict can never become $\top$ ($\bot$).

\begin{lemma}
\label{lemma:foreverinconclusive}
    Let~$L \subseteq \TSigma^\omega$, $\delays$ be a set of delays, and $\rho^*=(\sigma^*_1,\tau^*_1),\ldots,(\sigma^*_m,\tau^*_m)$ a nonempty $\delays$-observation. Then, for all $t \ge \tau(\rho^*)$
    \begin{enumerate}
        \item $\Delta_\delays(L,\rho^*, t) \subsetneq \delays \cap \{(\lat, \jit) \mid \lat \le \tau^*_1\} $ implies there is no $ \rho^*_1\in\TSigma^*$ such that $ \evalfuncd{\delays}{L}{\rho^* \cdot_t \rho^*_1, t'} = \top $ for any $ t' \ge t+\tau(\rho^*_1)$, and

        \item $\Delta_\delays(\overline{L},\rho^*, t) \subsetneq \delays \cap \{(\lat, \jit) \mid \lat \le \tau^*_1\} $ implies there is no $ \rho^*_1\in\TSigma^*$ such that $ \evalfuncd{\delays}{L}{\rho^* \cdot_t \rho^*_1, t'} = \bot  $ for any $ t' \ge t+\tau(\rho^*_1)$.
    \end{enumerate}
    
 \end{lemma}

\begin{proof}
Let  $\Delta_\delays(L,\rho^*, t) \subsetneq \delays \cap \{(\lat, \jit) \mid \lat \le \tau^*_1\} $, i.e., there is a $(\lat,\jit)\in\delays$ with $\lat \le \tau^*_1$ and 
$(\lat,\jit)\notin \Delta_\delays(L,\rho^*, t)$.
Thus, Lemma~\ref{lemma_consdelayunion} and Lemma~\ref{lemma:delaymono} imply that $(\lat,\jit)\in \Delta_\delays(\overline{L},\rho^* \cdot_t \rho^*_1, t')$ for all $\rho^*_1$ with $t \ge \tau(\rho^*)$ and $ t' \ge t+\tau(\rho^*_1)$.
Hence, $\Delta_\delays(\overline{L},\rho^* \cdot_t \rho^*_1, t') \neq \emptyset$ for all such $\rho^*_1$ and all such $t'$.
Finally, Lemma~\ref{lemma:delay} yields $ \evalfuncd{\delays}{L}{\rho^* \cdot_t \rho^*_1, t'} \neq \top $ for all such $\rho^*_1$ and all such $t'$.

The second claim is proven by a dual argument (swapping $L$ with $\overline{L}$ and $\top$ with $\bot$).
\end{proof}
 
Note that $\Delta_\delays(L,\rho^*, t) \subsetneq \delays \cap \{(\lat, \jit) \mid \lat \le \tau^*_1\} $ and $\Delta_\delays(\overline{L},\rho^*, t) \subsetneq \delays \cap \{(\lat, \jit) \mid \lat \le \tau^*_1\} $ can both be true simultaneously (as in Example~\ref{ex:foreverinconclusive}). In this situation, we will under no future observation reach a conclusive verdict.
 

\section{Towards an Algorithm}
\label{sec:towards}

Typically, monitoring algorithms rely on automata-based techniques. 
To this end, first the specification and its complement are translated into suitable automata. 
Then one computes the set of states reachable by processing the observation and then checks whether from one of these states the automaton can still accept an infinite continuation. If this is the case for both automata, then the verdict is~$\unknown$, if it is only the case for the automaton for the specification, then the verdict is $\top$, and vice versa for the complement automaton and $\bot$.

We want to follow the same blueprint, but we need to make adjustments to handle delay. 
Intuitively, we need to compute all states that are reachable by possible ground-truths of a given observation. 
However, a ground-truth may contain more events than the observation, as some events may not yet have been observed due to delay. 
This complicates the construction of the set of reachable states, as an unbounded number of events may not yet have been observed. 

In the definition of $\GT_\delays$ (Definition~\ref{def:Consistency}) there is an implicit universal quantification over all possible sequences of such events that have not yet been observed (e.g., the last five events in $\rho$ in Fig.~\ref{fig_gtdef}).
We exploit the fact that the verdicts are defined with respect to all possible extensions~$\mu$ of a possible ground-truth (i.e., also a universal quantification over the~$\mu$'s) to \myquot{merge} the universal quantification over events that have not yet been observed into the universal quantification of the extension~$\mu$. 
Then, a possible ground-truth has exactly the same number of events as the observation (i.e., ground-truth and observation have equal length (EL)). We begin defining this restricted notion of possible ground-truth by strengthening Definition~\ref{def:Consistency}.

\begin{definition}[EL-Consistency]\label{def:Consistency-el}
    Let $\rho^*=(\sigma^*_1,\tau^*_1),\ldots,(\sigma^*_m,\tau^*_m)$ be a $\{(\lat,\jit)\}$-observation and $\rho=(\sigma_1,\tau_1), \ldots, (\sigma_n,\tau_n)$ be a finite timed word.
  We say that $\rho$ is \emph{EL-consistent} with $\rho^*$ at observation time~$t \in \nnreals$ under latency $\lat$ and jitter~$\jit$ if and only if $\rho$ is consistent with $\rho^*$ at $t$ under $\lat$ and $\jit$ and $m = n$.
   We denote the set of timed words~$\rho$ that are EL-consistent with an $\{(\lat,\jit)\}$-observation~$\rho^*$ at observation time $t$ under latency~$\lat$ and jitter~$\jit$ by $\GT^\strict_{\lat,\jit}(\rho^*,t)$ and define $\GT^\strict_\delays(\rho^*,t) = \bigcup_{(\lat,\jit)\in\delays}\GT^\strict_{\lat,\jit}(\rho^*,t)$.
\end{definition}

\begin{example}
Continuing Example~\ref{example_gtdef}, an EL-consistent ground-truth of the observation~$\rho^*$ in Fig.~\ref{fig_gtdef} has exactly four events corresponding to the four events in the observation. Thus, there cannot be any unobserved events between $t-(\lat+\jit)$ and $t$ in an EL-consistent ground-truth (e.g., the last five events of $\rho$ in Fig.~\ref{fig_gtdef}).
\end{example}

The following lemma relates the original definition of consistency with EL-consistency.

\begin{lemma}
\label{lemma_elGT}\hfill

\begin{enumerate}

    \item \label{lemma_elGT_case1}
    
Let $\rho^*$ be a $\{(\lat,\jit)\}$-observation,
let $t \ge \tau(\rho^*)$,
let $\rho \in \GT_{\lat,\jit}^\strict(\rho^*,t)$, and let $\rho'$ be a finite timed word with $\tau(\rho') \le t - \max(\tau(\rho), t - (\lat+\jit))$. Then, 
$\rho \cdot_{\max(\tau(\rho), t - (\lat+\jit))} \rho' \in \GT_{\lat,\jit}(\rho^*,t)$.

    \item \label{lemma_elGT_case2}
Let $\rho^*$ be a $\{(\lat,\jit)\}$-observation (say with $m$ letters), let $t \ge \tau(\rho^*)$,
let $\rho \in \GT_{\lat,\jit}(\rho^*,t)$, and let $\rho'$ be the prefix of $\rho$ with $m$ letters. Then, $\rho' \in \GT_{\lat,\jit}^\strict(\rho^*,t)$.
    
\end{enumerate}
\end{lemma}

\begin{proof}
\ref{lemma_elGT_case1}.)
We have to show that $\rho \cdot_{\max(\tau(\rho), t - (\lat+\jit))} \rho'$ is consistent with $\rho^*$ at $t$ under $\lat$ and $\jit$. This follows directly from the fact that all events in $\rho'$ have time points (in $\rho \cdot_{\max(\tau(\rho), t - (\lat+\jit))} \rho'$) in the interval~$[t-(\delta+\jit),t]$ and are therefore covered by the third requirement of the definition of consistency.

\ref{lemma_elGT_case2}.)
We need to show that $\rho'$ is EL-consistent with $\rho^*$ at $t$ under $\lat$ and $\jit$. By definition, $\rho'$ has the same length as $\rho^*$ and the first two requirements of the definition of consistency are satisfied, as $\rho$ is consistent with $\rho^*$ at $t$ under $\lat$ and $\jit$ and $\rho'$ is a prefix of $\rho$. Hence, it is EL-consistent, as the third requirement only refers to ground-truths that have more letters than the observation.
\end{proof}

Now, we present the revised verdict function using only EL ground-truths. Note that merging the unobserved events from the possible ground-truth~$\rho$ into the extension~$\mu$ requires changing the time instant at which we concatenate the ground-truth and the extension: $t - (\lat + \jit)$ is the earliest time point at which an event can occur that may not yet have been observed at time~$t$. Due to jitter however, there might also be events after $t - (\lat + \jit)$ that have been observed, which are in the possible ground-truth~$\rho$: the last such event happened at time $\tau(\rho)$. Hence, we need to concatenate at time point~$\max(\tau(\rho), t - (\lat + \jit))$.

\begin{definition}[Monitor verdicts under delay -- EL version]
\label{def:delayedmonitor-el}
Given $L \subseteq \TSigma^\omega$, a set~$\delays$ of delays, a $\delays$-observation~$\rho^* \in \TSigma^*$, and $t \ge \tau(\rho^*)$, the function $\evalfuncsymbolstrict_{\delays} \colon \pow{\TSigma^\omega}\rightarrow \TSigma^*  \times \nnreals \rightarrow \verdicts$ evaluates to the verdict
\[
\evalfuncdstrict{\delays}{L}{\rho^*,t} = \left.
  \begin{cases}
    \top & \text{if $\rho \cdot_{\max(\tau(\rho), t - (\lat + \jit))} \mu \in L $ for all $(\lat,\jit) \in \delays$, all $\rho \in \GT^\strict_{\lat,\jit}(\rho^*,t)$ and all $ \mu \in \TSigma^{\omega}$}, \\
    \bot & \text{if $\rho \cdot_{\max(\tau(\rho), t - (\lat + \jit))} \mu \notin L $ for all $(\lat,\jit) \in \delays$, all $\rho \in \GT^\strict_{\lat,\jit}(\rho^*,t)$ and all $ \mu \in \TSigma^{\omega}$}, \\
    \unknown & \text{otherwise}.
  \end{cases}
  \right.
\]
$\evalfuncdstrict{\delays}{L}{\rho^*,t}$ is undefined when $t < \tau(\rho^*)$.
\end{definition}

Next, we show that both verdict functions coincide.

\begin{lemma}
\label{lemma:verdictvselverdict}
$\evalfuncdstrict{\delays}{L}{\rho^*,t} = \evalfuncd{\delays}{L}{\rho^*,t}$ for all $L \subseteq \TSigma^\omega$, all sets~$\delays$ of delays, all $\delays$-observations~$\rho^*$, and all $t \ge \tau(\rho^*)$.
\end{lemma}

\begin{proof}
Let $\evalfuncd{\delays}{L}{\rho^*,t} = \top$. 
We show $\evalfuncdstrict{\delays}{L}{\rho^*,t} = \top$ by proving that we have $\decorate{\rho} \cdot_{\max(\tau(\decorate{\rho}), t - (\lat + \jit))} \decorate{\mu} \in L$ for all $\decorate{\rho} \in \GT^\strict_{\lat,\jit}(\rho^*,t)$ for some $(\lat,\jit)\in \delays$ and all $\decorate{\mu} = (\sigma_1, \tau_1),(\sigma_2, \tau_2), \ldots \in \TSigma^\omega$.

First, consider the case where $\tau(\decorate{\rho}) < t - (\lat +\jit)$. 
Let $n$ be maximal with $\tau_n \le \lat+\jit$ (this is well-defined due to time-divergence), let $\decorate{\rho_1} = (\sigma_1, \tau_1), \ldots, (\sigma_n, \tau_n)$ and $\decorate{\mu_2} = ((\sigma_{n+1}, \tau_{n+1}),(\sigma_{n+2}, \tau_{n+2}), \ldots)- (\lat + \jit)$.
Note that $\decorate{\mu_2}$ is well-defined as $\tau_{n+1}$ is, by the choice of $n$, greater than $\lat + \jit$.
Then, Lemma~\ref{lemma_elGT}.\ref{lemma_elGT_case1} yields that $\decorate{\rho} \cdot_{t - (\lat + \jit)} \decorate{\rho_1}$ is in $\GT_{\lat,\jit}(\rho^*,t)$ and Remark~\ref{remark_concat}.\ref{remark_concat_case3} yields
\[
\decorate{\rho} \cdot_{\max(\tau(\decorate{\rho}), t - (\lat + \jit))} \decorate{\mu} = \decorate{\rho} \cdot_{ t - (\lat + \jit)} \decorate{\mu} = (\decorate{\rho} \cdot_{t - (\lat + \jit)} \decorate{\rho_1} ) \cdot_t \decorate{\mu_2}.
\]
Therefore, $\decorate{\rho} \cdot_{\max(\tau(\decorate{\rho}), t - (\lat + \jit))} \decorate{\mu}$ is the concatenation of the possible ground-truth~$(\decorate{\rho} \cdot_{t - (\lat + \jit)} \decorate{\rho_1} )$ of $\rho^*$ and the  suffix~$\decorate{\mu_2}$. As we have $\rho \cdot_t \mu \in L $ for all $\rho \in \GT_{\delays}(\rho^*,t)$ and all $ \mu \in \TSigma^{\omega}$ (due to $\evalfuncd{\delays}{L}{\rho^*,t} = \top$), we conclude $\decorate{\rho} \cdot_{\max(\tau(\decorate{\rho}), t - (\lat + \jit))} \decorate{\mu} \in L$ as required.

Now, consider the case where $\tau(\decorate{\rho}) \ge t - (\lat +\jit)$. 
Note that we have $t - \tau(\decorate{\rho}) \ge 0$ due to $\decorate{\rho} \in \GT^\strict_{\lat,\jit}(\rho^*,t)$.
Hence, let $n$ be maximal with $\tau_n \le t - \tau(\rho')$ (again, this is well-defined due to time-divergence), let $\decorate{\rho_1} = (\sigma_1, \tau_1) \cdots (\sigma_n, \tau_n)$ and $\decorate{\mu_2} = ((\sigma_{n+1}, \tau_{n+1})(\sigma_{n+2}, \tau_{n+2}) \cdots)- (t - \tau(\decorate{\rho}))$.
Again, $\decorate{\mu_2}$ is well-defined due to the choice of $n$.
Then, Lemma~\ref{lemma_elGT}.\ref{lemma_elGT_case1} yields that $\decorate{\rho} \cdot_{\tau(\decorate{\rho})} \decorate{\rho_1}$ is in $\GT_{\lat,\jit}(\rho^*,t)$ and Remark~\ref{remark_concat}.\ref{remark_concat_case3} yields 
\[
\decorate{\rho} \cdot_{\max(\tau(\decorate{\rho}), t - (\lat + \jit))} \decorate{\mu} = \decorate{\rho} \cdot_{\tau(\decorate{\rho})} \decorate{\mu} = (\decorate{\rho} \cdot_{\tau(\decorate{\rho})} \decorate{\rho_1} ) \cdot_t \decorate{\mu_2}.
\]
As $\decorate{\rho} \cdot_{\max(\tau(\decorate{\rho}), t - (\lat + \jit))} \decorate{\mu} $ is the concatenation of a possible ground-truth of $\rho^*$ and an arbitrary suffix, it is again, as required, in $L$.

Using a dual argument (i.e., swapping $\top$ with $\bot$ and $L$ with $\overline{L}$), we can show that $\evalfuncd{\delays}{L}{\rho^*,t} = \bot$ implies $\evalfuncdstrict{\delays}{L}{\rho^*,t} = \bot$.

Now, we show that $\evalfuncdstrict{\delays}{L}{\rho^*,t} = \top$ implies $\evalfuncd{\delays}{L}{\rho^*,t} = \top$. A dual argument again shows that $\evalfuncdstrict{\delays}{L}{\rho^*,t} = \bot$ implies $\evalfuncd{\delays}{L}{\rho^*,t} = \bot$. 
This will then complete our proof, as both functions only have three elements in their codomain and we have shown that two of them have the same preimage w.r.t.\ both functions.

So, let $\evalfuncdstrict{\delays}{L}{\rho^*,t} = \top$. 
We show $\evalfuncd{\delays}{L}{\rho^*,t} = \top$ by showing $\rho \cdot_t \mu \in L$ for all $\rho = (\sigma_1,\tau_1),\ldots,(\sigma_n, \tau_n) \in \GT_{\delays}(\rho^*,t)$ and all $ \mu = (\sigma_1', \tau_1'),(\sigma_2',\tau_2'),\ldots \in \TSigma^{\omega}$.
By definition, there is a $(\lat,\jit) \in\delays$ such that $\rho \in \GT_{\lat,\jit}(\rho^*,t)$.

Let $\rho^*$ have $m$ letters. 
If $m = n$, then we also have $\rho \in \GT^\strict_{\delays}(\rho^*,t)$. 
We consider two cases: If $\tau(\rho) < t-(\lat+\jit)$, then an application of Remark~\ref{remark_concat}.\ref{remark_concat_case1} yields
\[
\rho \cdot_t \mu 
=\rho \cdot_{t - (\lat + \jit)} (\mu + (\lat+\jit))=\rho \cdot_{\max(\tau({\rho}), t - (\lat + \jit))} (\mu + (\lat+\jit)), 
\]
and if $\tau(\rho) \ge t - (\lat+\jit)$, then an application of Remark~\ref{remark_concat}.\ref{remark_concat_case1} yields
\[
\rho \cdot_t \mu = \rho \cdot_{\tau(\rho)} (\mu+ (t - \tau(\rho)))= \rho \cdot_{\max(\tau({\rho}), t - (\lat + \jit))} (\mu+ (t - \tau(\rho))),
\]
where $t - \tau(\rho)$ is nonnegative by definition of consistency.
Hence, in both cases,  $\rho\cdot_t \mu$ is the concatenation of a possible EL ground-truth of $\rho^*$ and an arbitrary suffix.
Due to $\evalfuncdstrict{\delays}{L}{\rho^*,t} = \top$, all concatenations~$\rho \cdot_{\max(\tau(\rho), t - (\lat + \jit))} \mu$ for $(\lat,\jit) \in \delays$, $\rho \in \GT^\strict_{\lat,\jit}(\rho^*,t)$, and $ \mu \in \TSigma^{\omega}$ are in $L$, which yields $\rho\cdot_t \mu \in L$.

It remains to consider the case where $n > m$, which we again split into two subcases. But first let us define $\rho_1=(\sigma_1,\tau_1)\cdots(\sigma_m, \tau_m)$ as well as $\rho_2 = (\sigma_{m+1},\tau_{m+1}) \cdots (\sigma_n, \tau_n)$. Lemma~\ref{lemma_elGT}.\ref{lemma_elGT_case2} yields $\rho_1 \in \GT^\strict_{\delays}(\rho^*,t)$.

First, consider the subcase where $\tau(\rho_1) < t - (\lat +\jit)$. By definition of consistency, $n > m$ implies $\tau_{m+1} +\lat +\jit \ge t$, and thus $\tau_{m+1} \ge t - (\lat + \jit)$ ($\dagger$). 
Hence, an application of Remark~\ref{remark_concat}.\ref{remark_concat_case2} yields
\[
\rho\cdot_t \mu =  
\rho_1 \cdot_{t - (\lat +\jit)} [ (\rho_2- (t - (\lat+\jit))) \cdot_{\lat+\jit} \mu ]
=\rho_1 \cdot_{\max(\tau(\rho_1), t - (\lat + \jit))} [ (\rho_2- (t - (\lat+\jit))) \cdot_{\lat+\jit} \mu ].
\]
Note that the first time point of $\rho_2$, $\tau_{m+1} $, is at least $(t - (\lat+\jit))$ as required, as $\tau_{m+1} \ge t - (\lat + \jit)$ (see $\dagger$). 
Hence, $\rho\cdot_t \mu$ is the concatenation of a possible EL ground-truth of $\rho^*$ and an arbitrary suffix and therefore in $L$.

Finally, consider the subcase where $\tau(\rho_1) \ge t - (\lat +\jit)$.
An application of Remark~\ref{remark_concat}.\ref{remark_concat_case2} yields
\[
\rho\cdot_t \mu =
\rho_1 \cdot_{\tau(\rho_1)} [(\rho_2- \tau(\rho_1)) \cdot_{t - \tau(\rho_1)} \mu ]
= \rho_1 \cdot_{\max(\tau(\rho_1), t - (\lat + \jit))} [(\rho_2- \tau(\rho_1)) \cdot_{t - \tau(\rho_1)} \mu ].
\]
Again, this is well-defined as we have $\tau_{m+1} \ge \tau(\rho_1)$ (as $\tau_{m+1}$ is the next time instant after $\tau_{m}=\tau(\rho_1)$ in $\rho$) and as  $t \ge \tau(\rho_1)$ by the definition of consistency.
Hence, $\rho\cdot_t \mu$ is again the concatenation of a possible EL ground-truth of $\rho^*$ and an arbitrary suffix and therefore in $L$.
\end{proof}

Next, we show that we can indeed make the definition of $\evalfuncsymbolstrict$ effective using automata-theoretic constructions. 
First, we formally capture the set of states that can be reached by processing the possible EL ground-truths of an observation. Let $\aut$ be a \tba. 
We write $(q_0, v_0) \xrightarrow{\rho}_\aut (q_n, v_n)$ for a finite timed word~$\rho = (\sigma, \tau) \in \TSigma^*$ to denote the existence of a finite sequence of states~$(q_0, v_0) \transition{1} (q_1,v_1) \transition{2} \cdots \transition{n} (q_n,v_n)$
 of $\aut$ where for all $1 \leq i \leq n$ there is a transition $(q_{i-1},q_{i},\sigma_{i},\lambda_i,g_i)$ of $\aut$ such that $v_{i}(x) = 0$ for all $x$ in $\lambda_i$ and $v_{i-1}(x) + (t_i - t_{i-1})$ otherwise, and $g$ is satisfied by the valuation $v_{i-1}+(t_{i} - t_{i-1})$, where we use $t_0 = 0$. 
Given a \tba~$\aut$, a set~$\delays$ of delays, a finite observed timed word~$\rho^* \in \TSigma^*$, and $t \ge \tau(\rho^*)$, 
we define
\begin{align*}
\terminal{\aut}{\delays}{\rho^*,t} = \{ &(q, v + \max(0,(t - (\tau(\rho) + \lat +\jit)))) \mid (q_0, v_0) \xrightarrow{\rho}_\aut (q, v) \text{ where }\\
&\quad (q_0, v_0) \text{ with } q_0 \in Q_0 \text{, } v_0(x)=0 \text{ for all $x \in C$, and } 
\rho \in \GT^\strict_{\lat,\jit}(\rho^*,t) \text{ for some } (\lat,\jit)\in\delays\}.
\end{align*}
We call this the reach-set of $\rho^*$ in $\aut$ at $t$ w.r.t.\ $\delays$.

Next, we define the set of states of a \tba from where it is possible to reach an accepting location infinitely many times in the future, i.e., those states from which an accepting run is possible.
This is useful, because if processing a finite timed word leads to such a state, then the timed word can be extended to an infinite one in the language of the automaton, a notion that underlies the definitions of the verdict functions.
Given a \tba $\aut = (Q, Q_0, \Sigma, C, \Delta, \mathcal{F})$, the set of states with nonempty language is
\[\nonempty{\aut} = \{ (q,v) \mid q \in Q, v \in C \rightarrow \nnreals \text{ s.t. } \lang{\aut, (q,v)} \neq \emptyset \}.\]
The set~$\nonempty{\aut}$ can be computed using a zone-based fixpoint algorithm \cite{GrosenKLZ22}.
Using these definitions, we can give an \emph{effective} definition of the verdict functions, which we show to be equivalent to the previous definitions and implementable.

In the following definition, $\tbaty$ denotes the set of all \tba.
\begin{definition}[Monitoring \tba]\label{def:delayedmonitor-aut}
 Given a \tba~$\aut$, a complement automaton~$\compautomaton$ (i.e., with $\lang{\compautomaton} = \TSigma^\omega \setminus \lang{\aut}$), a set~$\delays$ of delays, a $\delays$-observation~$\rho^* \in \TSigma^*$, and $t \ge \tau(\rho)$, ${\monitordelayssymbol \colon \tbaty \times \tbaty \rightarrow \TSigma^* \times \nnreals \rightarrow \verdicts}$ computes the verdict
\[
\monitordelays{\aut, \compautomaton}{\rho^*, t} = \left.
  \begin{cases}
    \top & \text{if } \terminal{\compautomaton}{\delays}{\rho^*,t} \cap \nonempty{\compautomaton} = \emptyset, \\
    \bot & \text{if } \terminal{\aut}{\delays}{\rho^*,t} \cap \nonempty{\aut} = \emptyset, \\
    \unknown & \text{otherwise.}
  \end{cases}
  \right.
\]
$\monitordelays{\aut, \compautomaton}{\rho^*, t}$ is undefined if $t < \tau(\rho)$.
\end{definition}

Next we show that this automata-based definition of monitoring is equal to the verdict functions defined above.

\begin{theorem}
\label{thm:elverdictvsautverdict}
$\monitordelays{\aut,\compautomaton}{\rho^*,t} = \evalfuncdstrict{\delays}{\lang{\aut}}{\rho^*,t}$ for all sets~$\delays$ of delays, all \tba~$\aut$ (and complement automata~$\compautomaton$), all $\delays$-observations~$\rho^*$, and all $t \ge \tau(\rho^*)$.
\end{theorem}

\begin{proof}
We will show that $\terminal{\aut'}{\delays}{\rho^*,t} \cap \nonempty{\aut'} $ is nonempty if and only if there exists a $\rho \in \GT^\strict_{\lat,\jit}(\rho^*,t)$ and a $\mu \in \TSigma^\omega$ with $\rho \cdot_{\max(\tau(\rho),(t - (\lat +\jit)))}\mu \in \lang{\aut'}$ for any \tba~$\aut'$.
Then we obtain 
\begin{itemize}
    \item $\monitordelays{\aut,\compautomaton}{\rho^*,t} = \top$ if and only if $\evalfuncdstrict{\delays}{\lang{\aut}}{\rho^*,t} = \top$ by instantiating the equivalence for $\aut' = \compautomaton$, and 
    \item $\monitordelays{\aut,\compautomaton}{\rho^*,t} = \bot$ if and only if $\evalfuncdstrict{\delays}{\lang{\aut}}{\rho^*,t} = \bot$ by instantiating the equivalence for $\aut' = \aut$.
\end{itemize}
This completes the proof, as both functions only have three elements in their codomain and we have shown that two of them have the same preimage  w.r.t.\ both functions.

So, let $\terminal{\aut'}{\delays}{\rho^*,t} \cap \nonempty{\aut'} \neq\emptyset$. Then, by definition, there is a state~$(q,v')$ of $\aut'$ such that
\begin{itemize}
    \item $(q_0, v_0) \xrightarrow{\rho}_{\aut'} (q, v)$ for some initial state~$(q_0, v_0)$ of $\aut'$, some $\rho \in \GT^\strict_{\lat,\jit}(\rho^*,t)$  for some $(\lat,\jit)\in\delays$, and $v' = v + \max(0,(t - (\tau(\rho) + \lat +\jit)))$, and
    \item  there is an accepting infinite run of $\aut'$ starting in $(q,v')$ that processes some $\mu \in \TSigma^\omega$.
\end{itemize}
These two runs can be combined into an accepting run of $\aut'$ that starts in $(q_0, v_0)$ and processes 
\[\rho \cdot_{\tau(\rho) + \max(0,(t - (\tau(\rho) + \lat +\jit)))} \mu = \rho \cdot_{\max(\tau(\rho),(t - (\lat +\jit)))} \mu,\]
which implies that it is in $\lang{\aut'}$ as required.

Conversely, let there be a $\rho \in  \GT^\strict_{\lat,\jit}(\rho^*,t)$ and a $\mu \in \TSigma^\omega$ with 
\[\mu \cdot_{\max(\tau(\rho),(t - (\lat +\jit)))} \mu \in \lang{\aut'}.\]
Then, there exists an accepting run of $\aut'$ starting in some initial state~$(q_0, v_0)$ that processes $\rho \cdot_{\max(\tau(\rho),(t - (\lat +\jit)))} \mu$. 
This run can be split into 
\begin{itemize}
      \item $(q_0, v_0) \xrightarrow{\rho}_{\aut'} (q, v)$ for some state~$(q,v)$ of $\aut'$ and
     \item  an accepting infinite run of $\aut'$ starting in $(q,v')$ that processes $\mu$, where 
     \[v'=  v + \max(0,(t - (\tau(\rho) + \lat +\jit))) .\] 
\end{itemize}
Hence, $(q, v') \in \terminal{\aut'}{\delays}{\rho^*,t} \cap \nonempty{\aut'}$, which is therefore, as required, nonempty. 
\end{proof}

Recall that $\nonempty{\aut}$ can be computed for any given \tba~$\aut$. Therefore, in the next section, we show how to calculate $\terminal{\aut}{\delays}{\rho^*,t}$ for a given \tba~$\aut$, set~$\delays$ of delays, observation~$\rho^*$, and time point~$t$ using a zone-based algorithm. This will then allow us to compute verdicts effectively.

\section{A Zone-Based Online Monitoring Algorithm}
\label{sec:algo}

In this section, we demonstrate how to compute the reach-set of $\rho^*$ in $\aut$ at $t$ w.r.t.\ $\delays$.
So far we have developed the theory with observations, latency, and jitter being reals.  Now, we are concerned with algorithms and thus assume all these quantities to be rationals. 
For the monitoring algorithm, we use -- as standard in analysing timed automata models -- symbolic states being pairs~$(q, Z)$ of locations and zones. A zone is a finite conjunction of constraints of the form~$x \sim t$ and $x - x' \sim t$ for clocks~$x,x'$, constants $t \in \nnrats$, and $\sim\ \in \{<, \leq, =, \geq, >\}$. 
Given two zones $Z$ and $Z'$ over a set $C$ of clocks, and a set of clocks $\lambda \subseteq C$, we define the following operations on zones (which can be efficiently implemented using the DBM data structure \cite{DBLP:conf/ac/BengtssonY03}):

\begin{itemize}
    \item $Z[\lambda] = \{v \mid \exists v' \models Z \: \text{ s.t. }  v(x) = 0 \text{ if } x \in \lambda \text{, otherwise } v(x) = v'(x) \}$
    \item $Z^\nearrow = \{v \mid \exists v' \models Z \text{ s.t. }  v = v' + d  \text{ for some } d \in \nnreals\}$ 
    \item $Z \land Z' = \{v \mid v \models Z \textit{ and } v \models Z'\}$
\end{itemize}

We can use these functions  to compute the successor states after an input. Given a \tba $\aut = (Q, Q_0, \Sigma, C, \Delta, \mathcal{F})$, a symbolic state $(q, Z)$, and a letter $a \in \Sigma$, we define \[
\post((q, Z), a) = \{ (q', Z') \mid (q,q',a,\lambda,g) \in \Delta, Z' = (Z^\nearrow \land g)[\lambda] \},\] as the set of states one can reach by taking an $a$-transition at some point in the future from $(q, Z)$. 
Using $\post$ we can compute the successor states of a timed input $(a, \tau) \in \Sigma \times \nnrats$ by extending the zones with an additional clock $time$ just recording time since system start. The set of successors of a symbolic state is 
\[\succc((q, Z), (a, \tau)) = \{(q', Z') \mid (q', Z'') \in \post((q, Z), a), Z' = Z'' \land time = \tau\}\] 
and the set of successors of a set of symbolic states~$S$ is \[\succc(S, (a, \tau)) = \bigcup_{(q', Z')\in S} \succc((q', Z'), (a, \tau)).\]

In handling delayed observations, we assume that the delay set has the form
\[\delays = \{(\lat, \jit) \mid \lat \in [\ell, u]\}\]
for given $l,u,\jit \in \nnrats$, i.e., the latency $\delta$ is bounded by an interval $[\ell, u] \subseteq \nnreals$ and the jitter is bounded by $\jit \in \nnrats$.

To represent the latency and thereby be able to reason about and indirectly store the latency bounds, we add a clock $etime$ representing the ``expected'' real time that an event generated just now could be observed by the monitor after having been delayed according to the latency. This allows us 
\begin{enumerate}
    \item to represent the actual latency as $etime-time$,
    \item to represent the initial knowledge about latencies by initializing $etime-time$ to the initially known bounds on latency, namely $etime-time \in [\ell,u]$ by setting $time$ to 0 and constraining $etime$ to $[\ell,u]$, and
    \item\label{item:refine} to refine our knowledge about the actual latency after having observed an event~$(\sigma^*, \tau^*)$ by then setting $etime$ to a value in $ [\tau^*-\jit,\tau^*]$.
\end{enumerate}
Consequently, we change the initial zones to include the latency bounds $\ell$ and $u$ as the differences between the clocks $etime$ and $time$. This way, $etime$ represents the expected time an event is observed at the monitor, given $\ell$ and $u$, and $time$ represents the actual time the event happened (at the system). The aforementioned refinement (see Item \ref{item:refine} above and Fig.\ \ref{fig:delay-annotated-zone-operations}) then permits to deduce actual latency ranges consistent with the specification (or its negation) from observation times of events.

In detail, this refinement of the $etime-time$ relation works as follows. 
Given a \tba $\aut$ extended with the clocks $time$ and $etime$, and an observation $(\sigma, \tau^*) \in \Sigma \times \nnrats$, the successors of $(q, Z)$ are 
\[
\succc_{d}((q, Z), (\sigma, \tau^*)) = \{(q', Z') \mid (q', Z'') \in \post((q, Z), \sigma), Z' = Z'' \land etime \leq \tau^*  \land etime \geq \tau^* - \jit\}\
\]
and the successors $\succc_{d}(S, (\sigma, \tau^*))$ of a set of symbolic states $S$ is equal to $\bigcup_{(q, Z)\in S} \succc_{d}((q, Z), (\sigma, \tau^*))$.

The online monitoring algorithm will essentially apply  $\succc_d$ repeatedly to update the reach-set, once for each new observation. 
Note that there is a slight mismatch, as $\succc_{d}$ is computed with the two auxiliary clocks $time$ and $etime$, which are not clocks of $\aut$.

The initial reach-set is given by the following zone $Z_0^{d}$ requiring all ordinary clocks of the TBA $\aut$ to be zero and with  $time$ and $etime$ satisfying  $etime - time \in [\ell, u]$. 
That is

\[
    Z_0^{d} \equiv
    \underbrace{\vphantom{\bigwedge_{ x \in C \cup \{\globaltime\}}}\etime - \globaltime \le u \land \globaltime - \etime \le -\ell\,}_{\etime-\globaltime \in [\ell,u]} \land 
     \underbrace{\bigwedge_{ x \in C \cup \{\globaltime\}} x = 0 }_{\text{$x_1,\ldots,x_{|C|} = 0$, $\globaltime = 0$}}.
\]

Given a fixed jitter bound $\jit$, we can now compute the reach-set after a sequence of observations under delay, where the latency is bounded in $[\ell,u]$.

\begin{lemma}
\label{lemma_impl}
 Given a \tba $\aut$, 
a delay set~$\delays = \{(\lat, \jit) \mid \lat \in [\ell, u]\}$ with $\ell, u, \jit \in \nnrats$,
a $\delays$-observation~$\rho^* = (\sigma_1, \tau^*_1),\ldots,(\sigma_n, \tau^*_n)$, and $t \in \nnrats $ with $t \geq \tau^*_n$, let
$S_0 = \{(q_0, Z_0^{d}) \mid q_0 \in Q_0\}$ and $S_i = \succc_d(S_{i-1}, (\sigma_i, \tau^*_i))$ for $i \in \{1,\ldots, n\}$.
Then, the reach-set $\terminal{\aut}{\delays}{\rho^*,t}$ is the projection of
 \[ \{(q', Z') \mid (q', Z'') \in S_n, Z' = Z''^\nearrow \land  \etime = t - \jit\}\]
 to the clocks of $\aut$ (obtained by removing all constraints on $time$ and $etime$).
\end{lemma}

\begin{proof}
Given the form of $\delays$ we can rewrite the definition of the reach-set to
\begin{align*}
\terminal{\aut}{\delays}{\rho^*,t} = \{ &(q, v + \max(0,(t - (\tau(\rho) + \lat +\jit)))) \mid (q_0, v_0) \xrightarrow{\rho}_\aut (q, v) \text{ where }\\
&\quad (q_0, v_0) \text{ with } q_0 \in Q_0 \text{, } v_0(x)=0 \text{ for all $c \in C$, and }  
\rho \in \GT^\strict_{\lat,\jit}(\rho^*,t) \text{ for some }  \delta\in[\ell,u]\}.
\end{align*}
Now, extending the transition relation $\xrightarrow{\rho}_\aut$ to clock valuations over the extended set of clocks $C\cup\{time, etime\}$, we may further reformulate the reach-set as follows:
\begin{align*}
\terminal{\aut}{\delays}{\rho^*,t} = \{ &(q_n, v^* + \max(0,(t - (v^*_n(etime)  +\jit)))) \mid (q_0, v^*_0) \xrightarrow{\rho}_\aut (q_n, v^*_n) \text{ where }\\
& \quad(q_0, v^*_0) \text{ with } q_0 \in Q_0 \text{, } v^*_0(x)=0 \text{ for all $c \in C$, and }  \\
&\quad v^*_0(etime)-v^*_0(time)\in [\ell,u] \text{ and } 
 v^*_i(etime)\leq \tau^*_i \land v^*_i(etime)\geq \tau^*_i-\jit \land \sigma_i=\sigma^*_i \text{ for } i\in\{0,\ldots,n\} \}.
\end{align*}
A key observation for the correctness of the above reformulation,  is that the extended clocks $etime$ and $time$ are \emph{not} modified by the TBA $\aut$.  That is  $v^*_i(etime)-v^*_i(time)=v^*_0(etime)-v^*_0(time)\in [\ell,u]$ for all $i=\{0,\ldots,n\}$. 

Now let $\terminal{\aut}{\delays,j}{\rho^*,t}$ for $j\in\{0,\ldots,n\}$ be defined as follows:
\begin{align*}
\terminal{\aut}{\delays,j}{\rho^*,t} = \{ &(q_j, v^*_j) \mid (q_0, v^*_0) \xrightarrow{\rho}_\aut (q_j, v^*_j) \text{ where }\\
&\quad (q_0, v^*_0) \text{ with } q_0 \in Q_0 \text{, } v^*_0(x)=0 \text{ for all $c \in C$, and }  \\
&\quad v^*_0(etime)-v^*_0(time)\in [\ell,u] \text{ and }  v^*_i(etime)\leq \tau^*_i \land v^*_i(etime)\geq \tau^*_i-\jit \land \sigma_i=\sigma^*_i \text{ for } i\in\{0,\ldots,j\} \}.
\end{align*}
Then clearly $\terminal{\aut}{\delays,0}{\rho^*,t}=\{(q_0, Z_0^{d}) \mid q_0 \in Q_0\}=S_0$ and 
$\terminal{\aut}{\delays,j}{\rho^*,t} = \succc_d(\terminal{\aut}{\delays,j-1}{\rho^*,t}, (\sigma_i, \tau_i))$ for $j \in \{1,\ldots, n\}$, using standard arguments for the correctness of symbolic exploration of timed automata (see, e.g., \cite{DBLP:conf/ac/BengtssonY03}). Finally, as $(t - (v^*_n(etime)  +\jit)))=((t-\jit)- v^*_n(etime))$ it follows that $\terminal{\aut}{\delays}{\rho^*,t}= \terminal{\aut}{\delays,n}{\rho^*,t}^\nearrow \land  \etime = t - \jit$.
\end{proof}

This lemma allows us to implement a monitoring algorithm by computing the reach-sets and intersecting them with the set of nonempty language states as described in Definition~\ref{def:delayedmonitor-aut} and proven correct in Theorem~\ref{thm:elverdictvsautverdict}.

\begin{theorem}
The function~$\evalfuncsymbol_{\delays}(\lang{\aut})$ is effectively computable for specifications given by \tba\ $\aut$ and $\overline{\aut}$ (for the complement), and $\delays = \{(\lat, \jit) \mid \lat \in [\ell, u]\}$ for $\ell, u, \jit \in \nnrats$.
\end{theorem}

The observation of events may lead to refinement of the difference between $\globaltime$ and $\etime$ as depicted in Fig.~\ref{fig:delay-annotated-zone-operations}. 
This captures the definition of the sets of consistent delays from Section~\ref{sec:MonitoringDelay}.

\begin{lemma}
    Given $\aut$, $\delays$, $\rho^*$, $t$, and $S_n$ as in Lemma~\ref{lemma_impl}, we can compute the set of consistent delays by looking at the bounds on $etime - time$:  $\Delta_\delays( L(\aut), \rho^*, t) = \{(\lat, \jit) \in \delays \mid S_n \models etime - time = \lat\}$.
\end{lemma}

This information can be used to decorate the $\unknown$ verdict, so that we can report a set of bounds on the latency for which we would provide a $\top$ or $\bot$ verdict.

\begin{figure*}
    \centering
    \scalebox{.97}{
    \input{figures/zone-based-one-way-jitter}}
    \caption{Illustration of a single zone in the $\succc_{d}$ computation (only $\globaltime$-$\etime$ plane depicted). Left: initial zone (in green) is diagonally extrapolated for time passage and then intersected with the guard of an edge. Middle: observing event $\sigma$ at time $\tau$. By restricting $etime$ to $[\tau-\jit, \tau]$, the clock $time$ is restricted to when the event could have occurred at the system. Right: computing the future zone we see that the bound on $time - etime$ is now stricter and thus the bounds for the consistent latencies are refined.}
    \label{fig:delay-annotated-zone-operations}
\end{figure*}

\begin{example}\label{ex:method}
Let us show an example of our algorithm for monitoring under delayed observation. 
Note that, for the sake of readability, we use sets of clock constraints instead of conjunctions of clock constraints when specifying zones.

Consider the property~$\varphi = F_{[0,10]}a \land G_{[0,20]} \neg b$ from Fig.~\ref{fig:example1}. The \tba\ accepting $\lang{\varphi}$ and $\lang{\neg\varphi}$ are shown in Fig.~\ref{fig:aut-example}.
The nonempty language states for $\aut_\varphi$ and $\aut_{\neg\varphi}$ are 
$\nonempty{\aut_\varphi} = \{(q_0, \{x \le 10\}), (q_1, \true), (\varphi, \true)\}$ and
$\nonempty{\aut_{\neg\varphi}} = \{(q_0, \true), (q_1, \{x \le 20\}), (\neg\varphi, \true)\}$.
Let us assume the latency is between 0 and 10, and the jitter is bounded by $0.2$. 
Now we compute the reach-sets $S_0$ (initial), $S_1$ (after $(a, 17.3)$), and $S_2$ (after $(b, 27.5)$) as
\begin{align*}
    S_0 = \{&(q_0, \{x = 0,\ etime \leq 10,\ (etime - x) \in [0, 10]\})\},\\
    S_1 = \{&(q_1, \{x \in [7.1, 10],\ etime \in [17.1, 17.3],\ (etime - x) \in [7.1, 10]\}),\\ &(\neg\varphi, \{x \in [10, 17.3],\ etime \in [17.1, 17.3],\ (etime - x) \le 7.3\})\}\text{, and}\\
    S_2 = \{&(\varphi, \{x \in (20, 20.4], etime \in [27.3, 27.5),\ (etime - x) \in [7.1, 7.5]\}),\\
    & (\neg\varphi, \{x \in [17.3, 27.5],\ etime \in [27.3, 27.5],\ (etime - x) \in [0, 10]\})\}.
\end{align*}
Note that we omit the clock $time$ and only look at $x$ and $etime$ since $time$ and $x$ always have the same constraints.

All reach-sets intersect with both sets of nonempty language states; thus, the verdict is $\unknown$. However, we can refine this verdict with knowledge about the consistent delays that change after each observation.
The jitter bound is fixed at $0.2$, but the bounds on the latency can be found in the clock constraints on the difference between $etime$ and $x$.
For $\bot$, the latency range remains $[0,10]$ in all reach-sets. For $\top$,
the consistent latency range is $[0, 10]$ in $S_0$, $[7.1, 10]$ in $S_1$, and it is $[7.1, 7.5)$ in $S_2$. This means that if the latency is outside $[7.1, 7.5)$, then the verdict is $\bot$.

On the other hand, for the observation~$\rho^*=(a, 17.3), (b, 27.1)$ from Example~\ref{example:verdicts} (and using the same latency and jitter bounds as above), we compute the reach-sets~$S_0$, $S_1$, and $S_2'$ where
\[
S_2' = \{(\neg\varphi, \{x \in [16.9, 27.1], etime \in [26.9, 27.1], (etime - x) \in [0, 10]\})\}
.
 \]
As $S_2'$ has an  empty intersection with $\nonempty{\aut_\varphi}$, the verdict is $\bot$.
\end{example}

\section{Testing under Delay}
\label{sec:testing}

After having tackled the problem of monitoring under delay, we now consider the problem of testing under delay.
Testing is distinguished from passive monitoring by providing the system under observation with stimuli that are actively generated by the test harness. We therefore assume that the event alphabet~$\Sigma$ of our system is partitioned into an input alphabet~$\Sigma_I$ and an output alphabet~$\Sigma_O$ (where we take the point of view of the system under observation, i.e., events in $\Sigma_I$ are inputs to the system and events in $\Sigma_O$ are outputs of the system). 
The output channel behaves as in the case of passive monitoring, namely delivering messages from the system under test to the monitor with a delay given by a latency~$\lat_O$ and a jitter~$\jit_O$. The  input channel behaves symmetrically in that it delivers messages from the monitor to the system under test with again a (possibly different) delay given by a latency~$\lat_I$ and a jitter~$\jit_I$.
Our obligation then is to check satisfaction or violation of a specification $L$ over $\Sigma$ by a system while providing stimuli (from $\Sigma_I$) over a delayed input channel and observing reactions (from $\Sigma_O$) over a delayed output channel. Both input delay and output delay are unknown up to given bounds. 

Figure~\ref{fig:testing} illustrates the concept where inputs are observed by the test monitor before they actually arrive at the system under test, and outputs are observed by the monitor later than they are generated by the system under test.
Due to the uncertainty about delays in the two directions, it is obvious that the observed order of events may be different from the order in which the events occurred on the system side. For example, in Figure~\ref{fig:testing}, if $\lat_I = \lat_O = 10$ and $\jit_I = \jit_O = 1$ then the input event~$a$ would arrive at the system under test between $16.0$ and $17.0$, thus \emph{after} the output event~$b$ observed at $19.1$, which would have to be generated between $8.1$ and $9.1$ to be observed at $19.1$.
\begin{figure*}
    \centering
    \includegraphics[width=.75\linewidth]{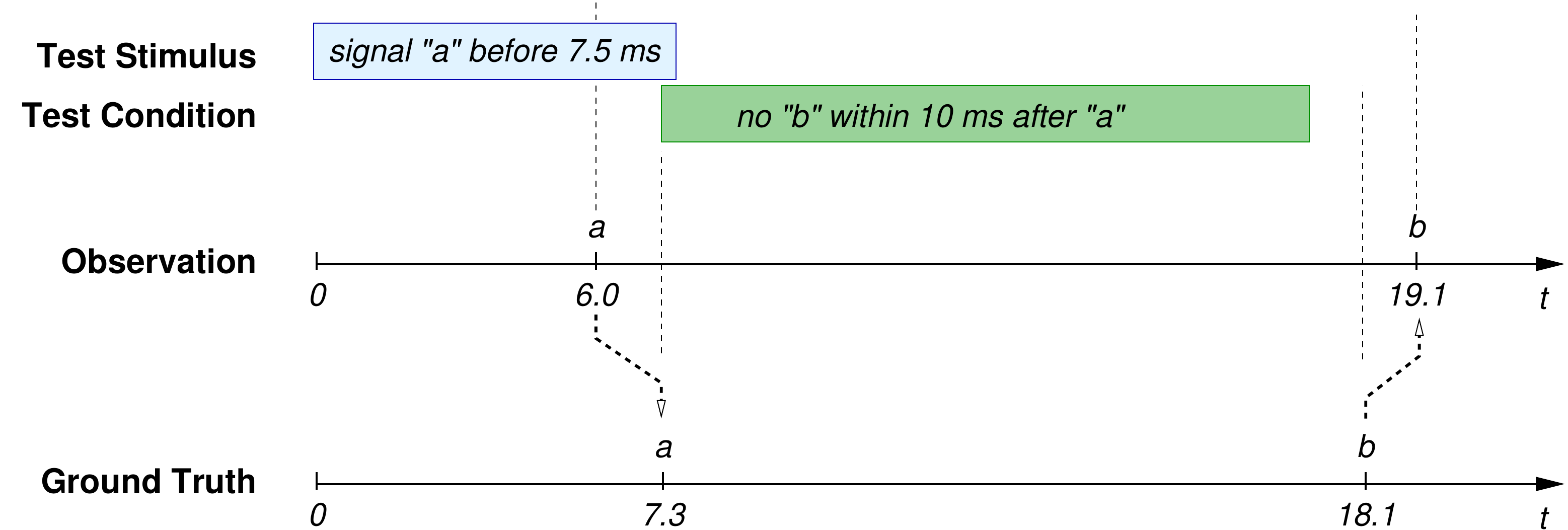}
    \caption{Executing real-time tests via delayed channels: both the stimuli and the responses are mediated via delayed channels. Input and output delays to/from the test harness are unknown up to bounds. Assuming latency bounded by $\infty > \lat_I \ge \lat_O \ge 0$ and jitter of at most $0.05$, a definite test verdict to the test specification $\left(F_{[0,7.5)}a\right) \rightarrow \left((\neg a) U (a \land G_{[0,10]}\neg b)\right)$, where $a$ is an input to the system under test and $b$ a response, is to be issued at time 19.1. Its value is $\top$, i.e., ``passed'', as either $\lat_I \ge 1.5$, in which case the antecedent of the implication is violated, or $1.5 > \lat_I \ge \lat_O$, in which case $a$ was received strictly before $6.0+1.5+0.05 = 7.55$ while the observed $b$ happened strictly after $19.1-1.5-0.05 = 17.55$.}
    \label{fig:testing}
\end{figure*}
In our definitions, we will admit such overtaking between messages in the input channels and the output channels, but will ---for simplicity of the exposition--- assume in-order delivery of input messages, as well as in-order delivery of output messages (see Footnote~\ref{footnote:overtake}), but will later restrict the setting when presenting our algorithm implementing testing under delay.

Due to overtaking only being allowed between inputs and outputs, the ground-truth corresponding to a monitor-side observation can be defined in terms of the two projections of the observed timed trace to the inputs and to the outputs, respectively.
To simplify our notation, given a (finite or infinite) timed word~$\rho$ over $\Sigma$, we denote its $\Sigma_I$ projection, i.e., the subsequence of pairs in $\Sigma_I \times \nnreals$, as $\rho\proj{I}$.
Similarly, we denote $\rho$'s $\Sigma_O$ projection as $\rho\proj{O}$.
Both projections are timed words over the respective alphabet.

As for monitoring, we begin by formalizing delay sets and observations that can be made under delay. 

\begin{definition}
An $IO$ delay set~$\delays$ is a nonempty subset of $\nnreals^4$ containing pairs of latencies and jitters for the input channel and pairs of latencies and jitters for the output channel.   
A $\delays$-observation, i.e.\ an observation that can in principle be made under delay in $\delays$, is a finite timed word~$\rho^*=(\sigma^*_1,\tau^*_1),\ldots,(\sigma^*_m,\tau^*_m)$ such that there is a $(\lat_I,\jit_I, \lat_O,\jit_O,) \in\delays$ with $\tau_i^* \ge \lat_O$, where $i$ is minimal with $\sigma^*_i \in \Sigma_O$. 
\end{definition}

Note that we only need to constrain occurrences of outputs in the observation, as inputs can be send to the system under test at any time.

Next, we formalize the notion of ground-truth, i.e., executions of the system under test that are consistent with an observation. 

\begin{definition}[Testing consistency]\label{def:Testingconsistency}
Let $\rho^*$ be a $\{(\lat_I,\jit_I,\lat_O,\jit_O)\}$-observation with $\rho^*\proj{I} = (\sigma^*_{I,1},\tau^*_{I,1}),\ldots,(\sigma^*_{I,m_I},\tau^*_{I,m_I})$ and $\rho^*\proj{O} = (\sigma^*_{O,1},\tau^*_{O,1}),\ldots,(\sigma^*_{O,m_O},\tau^*_{O,m_O})$.
Let $\rho$ be a finite timed word with $\rho\proj{I} = (\sigma_{I,1},\tau_{I,1}),\ldots,(\sigma_{I,n_I},\tau_{I,n_I})$ and $\rho\proj{O} = (\sigma_{O,1},\tau_{O,1}),\ldots,(\sigma_{O,n_O},\tau_{O,n_O})$.

We say that $\rho$ is \emph{(testing) consistent} with $\rho^*$ at observation time~$t \in \nnreals$ under $(\lat_I,\jit_I,\lat_O,\jit_O)$ if and only if
    \begin{enumerate}
   
        \item $\tau(\rho^*) \le t$ and $\tau(\rho) \le \max(t,\tau(\rho^*\proj{I}) + (\lat_I + \jit_I))$,
        
        \item[(2I)] $n_I = m_I$, and $\sigma_{I,i} = \sigma^*_{I,i}$ and $\tau_{I,i} -\tau^*_{I,i} \in [\lat_I,\lat_I+\jit_I]$ for all $i \in \{1,\ldots,m_I\}$,
        
        \item[(2O)] $n_O \ge m_O$, and $\sigma_{O,i} = \sigma^*_{O,i}$ and $\tau_{O,i}^* -\tau_{O,i} \in [\lat_O,\lat_O+\jit_O]$ for all $i \in \{1,\ldots,m_O\}$, and
        
        \item[(3)] if $ n_O > m_O $ then $  \tau_{m_O+1} \ge t - (\lat_O + \jit_O )$.
   
   \end{enumerate}
   We denote the set of timed words $\rho$ that are consistent with a $\{(\lat_I,\jit_I,\lat_O,\jit_O)\}$-ob\-ser\-vation~$\rho^*$ at observation time~$t$ under $(\lat_I,\jit_I,\lat_O,\jit_O)$ by $\testGT_{(\lat_I,\jit_I,\lat_O,\jit_O)}(\rho^*,t)$.
   Then, we define $\testGT_\delays(\rho^*,t)=\bigcup_{((\lat_I,\jit_I,\lat_O,\jit_O))\in\delays}\testGT_{(\lat_I,\jit_I,\lat_O,\jit_O)}(\rho^*,t)$.
\end{definition}

Before we continue, let us compare the previous definition to Definition~\ref{def:Consistency}, its analogue for monitoring. 
In Condition~(1), we must allow the ground-truth to have a longer duration than the observation to account for inputs that arrive at the system under test after time~$t$ due to the delay. 
All such inputs must have arrived at time~$\tau(\rho^*\proj{I}) + (\lat_I + \jit_I)$.
Furthermore, we can take into account that we are at time point~$t$, so we use $\max(t,\tau(\rho^*\proj{I}) + (\lat_I + \jit_I))$ as upper bound on the duration of the ground-truth.
Further, condition (2) of definition~\ref{def:Consistency} has to be stated for both channels independently. 
In Condition~(2I), we need to consider the difference~$\tau_{I,i} -\tau^*_{I,i}$ to account for the fact that inputs are sent from the test harness and arrive, after some delay, at the system under test. Also, we only consider ground-truths that have exactly the same number of inputs as the observation has (any additional input will be covered by extending the ground-truth to an infinite word). 
On the other hand, Condition~(2O) and Condition~(3) are similar to their counterparts in Definition~\ref{def:Consistency}.

\begin{example}

\label{example_testinggtdef}
Fig.~\ref{fig_testinggtdef} shows a $\{(\lat_I,\jit_I,\lat_O,\jit_O)\}$-observation and a consistent ground-truth and illustrates how the delay shifts the timestamps of the events. Here, $i$ is an input in $\Sigma_I$ and $o$ is an output in $\Sigma_O$. The length of $\rho$ is $n=10$ and the length of $\rho^*$ is $m=5$. Recall that $t$ is the time of observation.

\begin{figure}
\centering
    \begin{tikzpicture}[thick,xscale=1.25]
    \def\yunit{.1}

    \draw (0,0) -- (8.5,0);
    \draw (0,-2) -- (10,-2);

    \node[anchor = east] at (-.2,1) {time};
    \node[anchor = east] at (-.2,0) {$\rho^*$};    
    \node[anchor = east] at (-.2,-2) {$\rho$};

    
    \draw[thin, densely dotted, red] (0,.75) -- (0,-2);
    \node[] at (0,1) {$0$};

    
    \node[] at (7,1) {$t - (\lat_O + \jit_O)$}; 
    \draw[thin, densely dotted, red] (7,.75) -- (7,-2);

    \node[] at (8.5,1) {$t$};    
    \draw[thin, densely dotted, red] (8.5,.75) -- (8.5,-2);

    \node[] at (10,1) {$t + (\lat_I + \jit_I)$}; 
    \draw[thin, densely dotted, red] (10,.75) -- (10,-2);

    \draw[thin, densely dotted, red] (1.5,.75) -- (1.5,-2.75);
    \node[anchor = south,fill=white, inner sep=0] at (1.5,2.5*\yunit) {o};
    \node[] at (1.5,1) {$\tau_{O,1}^*$};
    \draw (1.5,\yunit) -- (1.5,-\yunit);

    \draw[thin, densely dotted, red] (2,.75) -- (2,-2.75);
    \node[anchor = south,fill=white, inner sep=0] at (2,2.5*\yunit) {i};
    \node[] at (2,1) {$\tau_{I,1}^*$};
    \draw (2,\yunit) -- (2,-\yunit);

    \draw[thin, densely dotted, red] (4,.75) -- (4,-2.75);
    \node[anchor = south,fill=white, inner sep=0] at (4,2.5*\yunit) {i};
    \node[] at (4,1) {$\tau_{I,2}^*$};
    \draw (4,\yunit) -- (4,-\yunit);

    \draw[thin, densely dotted, red] (6,.75) -- (6,-2.75);
    \node[anchor = south,fill=white, inner sep=0] at (6,2.5*\yunit) {o};
    \node[] at (6,1) {$\tau_{O,2}^*$};
    \draw (6,\yunit) -- (6,-\yunit);

    \draw[thin, densely dotted, red] (8,.75) -- (8,-2.75);
    \node[anchor = south,fill=white, inner sep=0] at (8,2.5*\yunit) {i};
    \node[] at (8,1) {$\tau_{I,3}^*$};
    \draw (8,\yunit) -- (8,-\yunit);

    \draw[thin, densely dotted, red] (0.5,.75) -- (0.5,-2.75);
    \node[anchor = south,fill=white, inner sep=0] at (0.5,-2-5*\yunit) {o};
    \node[] at (0.5,1) {$\tau_{O,1}$};
    \draw (0.5,-2+\yunit) -- (0.5,-2-\yunit);
    
    \draw[thin, densely dotted, red] (3,.75) -- (3,-2.75);
    \node[anchor = south,fill=white, inner sep=0] at (3,-2-5*\yunit) {i};
    \node[] at (3,1) {$\tau_{I,1}$};
    \draw (3,-2+\yunit) -- (3,-2-\yunit);
    
    \draw[thin, densely dotted, red] (4.5,.75) -- (4.5,-2.75);
    \node[anchor = south,fill=white, inner sep=0] at (4.5,-2-5*\yunit) {o};
    \node[] at (4.5,1) {$\tau_{O,2}$};
    \draw (4.5,-2+\yunit) -- (4.5,-2-\yunit);
    
    \draw[thin, densely dotted, red] (5,.75) -- (5,-2.75);
    \node[anchor = south,fill=white, inner sep=0] at (5,-2-5*\yunit) {i};
    \node[] at (5,1) {$\tau_{I,2}$};
    \draw (5,-2+\yunit) -- (5,-2-\yunit);

    \draw[thin, densely dotted, red] (9,.75) -- (9,-2.75);
    \node[anchor = south,fill=white, inner sep=0] at (9,-2-5*\yunit) {i};
    \node[] at (9,1) {$\tau_{I,3}$};
    \draw (9,-2+\yunit) -- (9,-2-\yunit);

    \draw[loosely dashed,stealth-, shorten >=3pt, shorten <=3pt] (1.5,0) -- (.5,-2);
    \draw[loosely dashed,stealth-, shorten >=3pt, shorten <=3pt]  (3,-2) -- (2,0);
    \draw[loosely dashed,stealth-, shorten >=3pt, shorten <=3pt] (6,0) -- (4.5,-2);
    \draw[loosely dashed,stealth-, shorten >=3pt, shorten <=3pt] (5,-2) -- (4,0);
    \draw[loosely dashed,stealth-, shorten >=3pt, shorten <=3pt] (9,-2) -- (8,0);
    




    \node[anchor = south,fill=white, inner sep=0] at (7.5,-2-5*\yunit) {o};
    \draw (7.5,-2+\yunit) -- (7.5,-2-\yunit);
    
    \node[anchor = south,fill=white, inner sep=0] at (8,-2-5*\yunit) {o};
    \draw (8,-2+\yunit) -- (8,-2-\yunit);
    
    \node[anchor = south,fill=white, inner sep=0] at (8.25,-2-5*\yunit) {o};
    \draw (8.25,-2+\yunit) -- (8.25,-2-\yunit);
    
    \node[anchor = south,fill=white, inner sep=0] at (9.5,-2-5*\yunit) {o};
    \draw (9.5,-2+\yunit) -- (9.5,-2-\yunit);
    
    \node[anchor = south,fill=white, inner sep=0] at (9.75,-2-5*\yunit) {o};
    \draw (9.75,-2+\yunit) -- (9.75,-2-\yunit);

    \end{tikzpicture}
    \caption{A $\{(\lat,\jit)\}$-observation~$\rho^*$ and a consistent ground-truth~$\rho$.}
    \label{fig_testinggtdef}
\end{figure}

In particular, notice the following:
\begin{itemize}
    \item As in the case of testing, no output can occur in the observation~$\rho^*$ with a timestamp smaller than $\lat_O$, as it takes at least $\lat_O$ units of time for an output to be send from the system through the output channel to the test harness. Obviously, at the system side (i.e., in the ground-truth~$\rho$) events can happen at any timestamp, also before $\lat_O$ (e.g., the first $o$). Similarly, inputs can be send by the test harness at any time. 

    \item The difference~$\tau_{O,j}^*-\tau_{O,j}$ for $j \in \{1,2 \}$ (i.e., the difference between the time an output is observed at the test harness and the time the event was emitted by the system) must be in the interval~$[\lat_O,\lat_O+\jit_O]$. Dually, the difference~$\tau_{I,j}-\tau_{I,j}^*$ for $j \in \{1,2,3 \}$ (i.e., the difference between the time an input is emitted by the test harness and the time the event was received by the system) must be in the interval~$[\lat_I,\lat_I+\jit_I]$.

    This can lead to overtaking between inputs and outputs, as evident by the second output being sent from the system before the second input arrives (i.e., $\tau_{O,2} < \tau_{I,2}$). At the test harness, the situation is reversed: Here, the second input is sent before the second output is received (i.e., $\tau_{O,2}^* > \tau_{I,2}^*$).

    \item As in the case of monitoring, there may be outputs in the ground-truth that have not been observed by the test harness, e.g., the last five outputs in $\rho$. As before, such events can only have timestamps in the interval~$[t-(\lat_O+\jit_O), \tau(\rho)]$, as all earlier events must necessarily have been observed at the test harness at time~$t$. 
    Dually, an input sent by the test harness before time~$t$ may arrive after $t$, e.g., the last input. However, all such events must have arrived at the system at time~$t + (\lat_I + \jit_I)$.

\end{itemize}
\end{example}

\begin{remark}
\label{remark_testgtnonempty}
By extending the argument presented in Remark~\ref{remark_gtnonempty}, one can show that $\testGT_\delays(\rho^*,t)$ is always nonempty, if $\rho^*$ is a $\delays$-observation and $t \ge \tau(\rho^*)$.
\end{remark}

Using these ground-truths, we can now define testing under delay. 

\begin{definition}[Testing verdicts under delay]
\label{def:delayedtester}
Given a language $L \subseteq \TSigma^\omega$, a set of possible observation delays~$\delays$, a $\delays$-observation~$\rho^* \in \TSigma^*$, and an observation time~$t \ge \tau(\rho^*)$, the function $\testevalfuncsymbol_{\delays} \colon \pow{\TSigma^\omega}\rightarrow \TSigma^*  \times \nnreals \rightarrow \verdicts$ evaluates to the verdict
\[
\testevalfuncd{\delays}{L}{\rho^*,t} = \left.
  \begin{cases}
    \top & \text{if } \rho \cdot_{\max(t,\tau(\rho))} \mu \in L $ for all $\rho \in \testGT_{\delays}(\rho^*,t)$, and all $ \mu \in \TSigma^{\omega}, \\
    \bot & \text{if } \rho \cdot_{\max(t,\tau(\rho))} \mu \notin L $ for all $\rho \in \testGT_{\delays}(\rho^*,t)$, and all $ \mu \in \TSigma^{\omega}, \\

    \unknown & \text{otherwise}.
  \end{cases}
  \right.
\]
$\testevalfuncd{\delays}{L}{\rho^*,t}$ is undefined when $t < \tau(\rho^*)$.
\end{definition}

Similarly to the corresponding result for monitoring (see Lemma~\ref{lemma:testdelaycertainty}), one can prove that our testing verdict function is monotone in the delays: decreasing the uncertainty about the possible delays preserves conclusive verdicts.

\begin{lemma}
\label{lemma:testdelaycertainty}
Let $L \subseteq \TSigma^\omega$, $\rho^* \in \TSigma^*$, let $\delays \subseteq \delays'$ be delay sets, let $\rho^*$ be a $\delays$-observation, and let $t \ge \tau(\rho^*)$. 
Then, $\testevalfuncd{\delays'}{L}{\rho^*,t}=\top$ implies $\testevalfuncd{\delays}{L}{\rho^*,t}=\top$ and $\testevalfuncd{\delays'}{L}{\rho^*,t}=\bot$ implies $\testevalfuncd{\delays}{L}{\rho^*,t}=\bot$.
\end{lemma}

Also, as for monitoring, we can refine the testing verdict function from Definition~\ref{def:delayedtester} to provide information about the delay parameters~$(\lat_I,\jit_I,\lat_O,\jit_O)$ that can explain an observation.  
Given~$L \subseteq \TSigma^\omega$, a finite timed word~$\rho^* \in \TSigma^*$, and $t \ge \tau(\rho^*)$, the set of delays~$\testDelta(L,\rho^*,t)$ that are consistent with the observation $\rho^*$ at $t$ is defined as
\[ \testDelta(L,\rho^*,t) = \{ (\lat_I,\jit_I,\lat_O,\jit_O) \mid 
    \exists \rho\in\testGT_{\lat_I,\jit_I,\lat_O,\jit_O}(\rho^*,t)\, \exists \mu\in\TSigma^{\omega} \text{ s.t. }\rho \cdot_{\max(t,\tau(\rho))} \mu \in L\}. \]
We denote by $\testDelta_{\delays}(L,\rho^*,t)$ the set $\testDelta(L,\rho^*,t) \cap \delays$.

Conclusive testing verdicts can again be characterized via the sets of consistent delays, which can be proven by the same argument as for the special case of monitoring (see Lemma~\ref{lemma:delay}).

\begin{lemma}
\label{lemma:testdelay}
Given $L \subseteq \TSigma^\omega$, a set~$\delays$ of delays, a $\delays$-observation~$\rho^* \in \TSigma^*$, and $t \ge \tau(\rho^*)$, we have
    \begin{enumerate}
       \item  $\testDelta_{\delays}(L,\rho^*,t)=\emptyset$ if and only if $\testevalfuncd{\delays}{L}{\rho^*,t}=\bot$, and 
       \item  $\testDelta_{\delays}(\overline{L},\rho^*,t)=\emptyset$ if and only if $\testevalfuncd{\delays}{L}{\rho^*,t}=\top$.
    \end{enumerate}
\end{lemma}

Again, even in the case when both delay-sets are nonempty (i.e., the verdict is $\unknown$), we can still provide useful information in terms of the sets~$\Delta(L, \rho^*, t)$ and $\Delta(\overline{L}, \rho^*, t)$ of consistent delays. 
They are non-increasing during observations: By extending observations, we (potentially) reduce the set of consistent delays.
This result is again proven along the lines of the proof of the similar result for monitoring (Lemma~\ref{lemma:delaymono}).

\begin{lemma}
\label{lemma:testdelaymono}
Let $(\rho_1^*,t_1)\sqsubseteq (\rho_2^*,t_2)$ for finite timed words~$\rho_1^*$ and $\rho_2^*$ with $t_1 \ge \tau(\rho_1^*)$ and $t_2 \ge \tau(\rho_2^*)$ and $t_2 \ge t_1$. Then, $\testDelta(L,\rho_1^*,t_1)\supseteq\testDelta(L,\rho_2^*,t_2) $. 
\end{lemma}

Now, we show how to compute $\testevalfuncsymbol_{\delays}(L)$ following the blueprint developed for monitoring under delay:
Given $\rho^*$ and $t$, compute the reach-set for $ \testGT_{\delays}(\rho^*,t)$ in the two \tba for $L$ and the complement of $L$ and then intersect with the respective sets of nonempty language states. 
Note that we cannot work with equal-length ground-truths here, as we have done in the case of monitoring, since we need to allow for unobserved outputs that were sent before the last input from the observation has arrived at the system. In Figure~\ref{fig_testinggtdef}), consider the three outputs in the interval~$[t-(\lat_O+\jit_O), t]$: As the input at time~$\tau_{I,3}$ needs to be part of the ground-truth, we need to allow these outputs as well.

Hence, we need to update our definition of reach-set to consider arbitrary ground-truths, not just equal-length ones:
Given a \tba~$\aut$, a set~$\delays$ of delays, a finite observed timed word~$\rho^* \in \TSigma^*$, and $t \ge \tau(\rho^*)$, 
we define
\begin{align*}
\terminalprime{\aut}{\delays}{\rho^*,t} = \{ &(q, v + \max(0,(t - \tau(\rho) ))) \mid (q_0, v_0) \xrightarrow{\rho}_\aut (q, v) \text{ where }\\
&\quad (q_0, v_0) \text{ with } q_0 \in Q_0 \text{, } v_0(x)=0 \text{ for all $x \in C$, and } 
\rho \in \testGT_{\delays}(\rho^*,t)\}.
\end{align*}
We call this the reach-set of $\rho^*$ in $\aut$ at $t$ w.r.t.\ $\delays$.
Using this, we can define the automata-based verdict function.

\begin{definition}\label{def:delayedtesting-aut}
 Given a \tba~$\aut$, a complement automaton~$\compautomaton$ (i.e., with $\lang{\compautomaton} = \TSigma^\omega \setminus \lang{\aut}$), a set~$\delays$ of delays, a $\delays$-observation~$\rho^* \in \TSigma^*$, and $t \ge \tau(\rho)$, ${\monitordelayssymbolprime \colon \tbaty \times \tbaty \rightarrow \TSigma^* \times \nnreals \rightarrow \verdicts}$ computes the verdict
\[
\monitordelaysprime{\aut, \compautomaton}{\rho^*, t} = \left.
  \begin{cases}
    \top & \text{if } \terminalprime{\compautomaton}{\delays}{\rho^*,t} \cap \nonempty{\compautomaton} = \emptyset, \\
    \bot & \text{if } \terminalprime{\aut}{\delays}{\rho^*,t} \cap \nonempty{\aut} = \emptyset, \\
    \unknown & \text{otherwise.}
  \end{cases}
  \right.
\]
$\monitordelaysprime{\aut, \compautomaton}{\rho^*, t}$ is undefined if $t < \tau(\rho)$.
\end{definition}

Using arguments similar to those to prove the analogous result for monitoring under delay (see Theorem~\ref{thm:elverdictvsautverdict}), one can show that this automata-based definition of testing is equal to the verdict function defined above.

\begin{theorem}
\label{thm:testingverdictvsautverdict}
$\monitordelaysprime{\aut,\compautomaton}{\rho^*,t} = \testevalfuncd{\delays}{\lang{\aut}}{\rho^*,t}$ for all sets~$\delays$ of delays, all \tba~$\aut$ (and complement automata~$\compautomaton$), all $\delays$-observations~$\rho^*$, and all $t \ge \tau(\rho^*)$.
\end{theorem}

This result, while theoretically appealing, is unfortunately hard to implement efficiently using our zone-based approach.

\begin{example}
Recall that the delay~$(\lat_I,\jit_I,\lat_O,\jit_O)=(0,2,0,0)$ expresses that the input channel has latency zero and jitter two while the output channel has latency and jitter zero.
Fix some $n>0$.
Now, consider a $(0,2,0,0)$-observation~$\rho^*$ with $n$ inputs in the interval~$[0,1]$ and $n$ outputs in the interval~$[1,2]$. 
The inputs arrive at the system in the interval~$[0,3]$ and all outputs are observed at the same time as they have been sent, i.e., also in the interval~$[1,2]$.

In the following, we are focusing only on ground-truths where all inputs arrive in the interval~$[1,2]$ as well. 
Each input can arrive at any time in the interval~$[1,2]$, as long as the order of the inputs is preserved.
Thus, the $n$ outputs induce $n+1$ buckets in which the $n$ inputs can be placed into in the ground-truth. 
There are $\binom{2n}{n}$ ways to do so, a quantity that grows exponentially in $n$.
Hence, even when abstracting away exact time points (this is what our zone-based construction does), we still need to account for all these possible orderings (which we would need to handle explicitly), leading to a combinatorial explosion.
\end{example}

Hence, in the following, we consider a setting that does not exhibit this combinatorial explosion while still being expressive enough to express properties from the use case studied  in Section~\ref{sec:implementation}.
Intuitively, we disallow the overtaking of inputs and outputs by considering only words in which inputs and outputs strictly alternate. Said differently, after each input, we need to wait for a corresponding output before the next input can be sent, and vice versa.

\begin{definition}
Let $\Sigma$ be the disjoint union of $\Sigma_I$ and $\Sigma_O$. A word~$(\sigma_1,\tau_1), (\sigma_2,\tau_2), \ldots$ is IO-alternating if $\sigma_j \in \Sigma_I$ for all odd $j$ and $\sigma_j \in \Sigma_O$ for all even $j$ (note that the word needs to start with an input). Given a language~$L \subseteq \TSigma^\omega$, let $IO(L)$ be the set of IO-alternating words in $L$.
\end{definition}

\begin{remark}\hfill
\label{remark_io}
\begin{enumerate}
    \item\label{remark_io_test}
    Given a $\tba$~$\aut$, one can effectively test whether it accepts only IO-alternating words (i.e., whether $\IO(L(\aut)) = L(\aut)$) by testing the intersection of $L(\aut)$ and the language of words with two consecutive inputs or two consecutive outputs for emptiness. As \tba are closed under intersection and emptiness is decidable~\cite{alur1994tba}, this can indeed be done effectively. 
    \item\label{remark_io_restr}
    Given a \tba~$\aut$, one can effectively construct a \tba that accepts $\IO(L(\aut))$ by taking the product of $\aut$ and a \tba that accepts the language of all IO-alternating words to accept the intersection of both languages.
\end{enumerate}
\end{remark}

In the following, we restrict ourselves only to words that are IO-alternating.
This requires to further restrict the observations that can be made:
An observation~$\rho^* = (\sigma_1,\tau_1) \cdots (\sigma_n,\tau_n)$ is an IO-observation (under $(\lat_I, \jit_I, \lat_O, \jit_O)$) if $\rho^*$ is IO-alternating and $\tau_{j+1} - \tau_{j} \ge \lat_I + \lat_O $ for all odd $j$, i.e., there is enough time between an input sent and the next observed output to sent these events through the corresponding channels.

The following theorem shows that we can test \mitl specifications w.r.t.\ IO-alternating words.
Note that $\IO(\lang{\aut}$ and $\IO(\lang{\overline{\aut}}$ partition the set of IO-alternating words, which implies that $\monitordelayssymbolprime(\IO(\lang{\aut}), \IO(\lang{\overline{\aut}}))$ is well-defined.
Furthermore, by considering these \tba, we restrict ourselves to extensions of ground-truths of the observation that are IO-alternating, all other words are ignored.

\begin{theorem}
The restriction of $\monitordelayssymbolprime(\IO(\lang{\aut}), \IO(\lang{\overline{\aut}}))$ to IO-observations is effectively computable for specifications given by \tba $\aut$ and $\overline{\aut}$, and $\delays = \{(\lat_I, \jit_I, \lat_O, \jit_O) \mid \lat_I \in [\ell_I, u_I] \text{ and } \lat_O \in [\ell_O, u_O]\}$ for given $\ell_I, u_I, \ell_O, u_O, \jit_I, \jit_O \in \nnrats$.
\end{theorem}

\begin{proof}
We generalize the construction presented above for monitoring to testing, i.e., we show how to compute reach-sets and then intersect them with the non-empty language states. 
Due to Theorem~\ref{thm:testingverdictvsautverdict}, this suffices to compute $\monitordelayssymbolprime(\IO(\lang{\aut}), \IO(\lang{\overline{\aut}}))$.

Instead of just adding the clock~$\etime$ as in the case of monitoring under delay, we now add two clocks $\etime_I, \etime_O$, and change the initial zones to include the latency bounds as the differences between the clocks $\etime_O$, $\etime_I$ and $\globaltime$. The idea here is that $\etime_I \le \globaltime \le \etime_O$ and that $\globaltime-\etime_I$ reflects the input delay from the test harness to the system under test that applies to test stimuli, while $\etime_O-\globaltime$ corresponds to the output delay affecting responses from the system under test.
Given a symbolic state $(q, Z)$ where $Z$ is extended with clocks $\globaltime, \etime_I$, and $\etime_O$, 
and a pair $(\square, \tau) \in \{I, O\} 
\times \nnrats$ we define the following zone operation:
\[
(q, Z) \land (\square, \tau) = 
\begin{cases} 
    (q, Z \land \etime_I \in [\tau, \tau+\jit_I]) &\text{if }\square = I, \\
    (q, Z \land \etime_O \in [\tau-\jit_O,\tau]) &\text{if }\square = O.
\end{cases}
\Bigg.
\]

Given a \tba~$\aut$ extended with clocks $\globaltime, \etime_I$ and $\etime_O$, and an observation $(\sigma, \tau) \in \Sigma \times \nnrats$, the set of possible successors of a symbolic state $(q, Z)$ is 
\begin{eqnarray*}
  Succ_{IO}((q, Z), (\sigma, \tau)) &=& \{(q', Z') \land (\dir(\sigma), \tau) \mid (q', Z') \in post((q, Z), \sigma)\}, \text{ where}\\
  \dir(\sigma) &=& \begin{cases}
      I & \text{if } \sigma \in \Sigma_I,\\
      O & \text{if }\sigma \in \Sigma_O,
  \end{cases}
\end{eqnarray*}

and 
\[Succ_{IO}(S, (\sigma, \tau)) = \bigcup_{(q, Z)\in S} succ_{IO}((q, Z), (\sigma, \tau))\] 
collects the successors of a set of symbolic states $S$. In order to use $Succ_{IO}$ to compute the reach-set, we need a different initial zone than $Z_0$. We define $Z_0^{IO}$ to extend $Z_0$ with $\globaltime$, $\etime_I$ and $\etime_O$ and set the bounds on the differences of these clocks to be the input/output delay bounds. Given a \tba~$\aut = (Q, Q_0, \Sigma, C, \Delta, \mathcal{F})$ and latency bounds $\ell_I, u_I, \ell_O, u_O$ then we have that:\footnote{Technically speaking, the constraints defining~$Z_0^{IO}$ require negative clock values for $\etime_I$. This can be avoided by a shift in the bounds between $\globaltime$ and $\etime_I$, or by letting time pass.}
\begin{equation*}
\begin{split}
         Z_0^{IO} \equiv& \overbrace{\etime_O - \globaltime \le u_O \land \globaltime - \etime_O \le -\ell_O\,}^{\etime_O-\globaltime \in [\ell_O,u_O]} \land \\ 
         &\forall x \in C \cup \{\globaltime\}: x = 0 \ \land\\
         &\underbrace{\etime_I - \globaltime \le -\ell_I \land \globaltime - \etime_I \le u_I\,}_{\globaltime-\etime_I \in [\ell_I, u_I]}.
\end{split}
\end{equation*}

We can now compute the reach-set after a sequence of observations.
Given a \tba $\aut$, a finite timed word $\rho = (\sigma_1, \tau_1),\ldots,(\sigma_n, \tau_n) \in \TSigma^*$, a time point~$t \ge \tau(\rho)$, and a set of delays $\delays = \{(\lat_I, \jit_I, \lat_O, \jit_O) \mid \lat_I \in [\ell_I, u_I] \text{ and } \lat_O \in [\ell_O, u_O]\}$ for given $\ell_I, u_I, \ell_O, u_O, \jit_I, \jit_O \in \nnrats$, 
the reach-set can be computed as 
$\terminalprime{\aut}{\delays}{\rho^*,t} = S_{n+1}$
where $S_0 = \{(q_0, Z_0^{IO}) \mid q_0 \in Q_0\}$ and $S_i = Succ_{IO}(S_{i-1}, (\sigma_i, \tau_i))$ for $i \in [1, n]$, and 
\[
S_{n+1} = 
\begin{cases}
    S_n \land (I, t) &\text{if } \dir(\sigma_n) = O\\
    S_n \land (O, t) &\text{if } \dir(\sigma_n) = I \textit{ and } t > l_I + u_O + \jit_O \\
    S_n & \text{otherwise}\\
\end{cases}
\]

The computation of the set of non-empty language states requires us to enforce IO-alternation by taking the intersection with an automaton $\aut_{IO}$ that accepts all IO-alternating words such that $
\IO(\lang{\aut}) = \lang{\aut \otimes \aut_{IO}}$ as described in Remark~\ref{remark_io}.\ref{remark_io_restr}.
This allows us to implement the testing function by computing the reach-sets and intersecting them with the set of non-empty language states.
\end{proof}

\lstdefinestyle{console}{
    frame=single,
    basicstyle=\ttfamily\scriptsize,
    breakatwhitespace=true,
    breaklines=true,
    captionpos=b,
    keepspaces=true,
    numbers=left,
    numbersep=5pt,
    showspaces=false,
    showstringspaces=false,
    showtabs=false,
    tabsize=4,
}

\section{Implementation}
\label{sec:implementation}

In this section we demonstrate our tool implementation of the monitoring and testing procedures described in this paper. We run experiments to demonstrate the efficiency and illustrate the evaluation of consistent latencies.

The tool \textsc{MoniTAal}\footnote{\url{https://github.com/DEIS-Tools/MoniTAal}} implements monitoring and testing under delay as described in this paper. This includes the difference-bounded matrix data structure to handle clock zones, parsing property automata modeled in \textsc{Uppaal}, 
computing the set of nonempty language states, computing the reach-sets in an online fashion over an observed word based on latency and jitter bounds in $[0, \infty[$, providing verdicts~$\top, \bot$ or $\unknown$, and providing bounds on the latency values that are consistent with $\top$ and $\bot$ for both inputs and outputs.

MoniTAal monitors a property by utilizing an automaton for the property and its complement. A series of observations can then be provided to compute verdicts. Bounds on latency and jitter can be given when monitoring under delay, where the valid latency values are computed for each verdict. 

We give a short demonstration going through Example~\ref{ex:method} using MoniTAal to monitor the property $F_{[0,10]}a \land G_{[0,20]} \neg b$ over the timed word $(a, 17.3), (b, 27.5)$, where the latency is between 0 and 10, and the jitter bound is $0.2$. Listing~\ref{lst:example} shows the output of the tool on lines 3-9 and 13-19 after each observation $(a,173)$ on line 1 and $(b, 271)$ on line 11. Note that we multiply all timing values by 10 in order to use integer, rather than rational, time points.
After the first observation, the verdict is $\unknown$ (inconclusive) (line 3). Furthermore, we see that the consistent latencies for the $\top$ (positive) verdict are now tightened to $[71, 100]$ (line 5). The consistent latencies for the $\bot$ (negative) verdict are still in $[0, 100]$ (line 8).
The second observation is $(b, 271)$ on line 11, after which the verdict is still $\unknown$ (inconclusive) (line 13), but the $\top$ consistent latencies are now tightened further to be within $[71,75)$ (line 15), while the $\bot$ consistent latencies are still in $[0, 100]$ (line 18).

\begin{lstlisting}[style=console, basicstyle=\ttfamily\small, label={lst:example}, caption={Demonstration of MoniTAal over Example~\ref{ex:method}.}]
Input: @173 a

Verdict: INCONCLUSIVE
Positive:
Consistent latencies: {[71,100]}
Jitter bound: 2
Negative:
Consistent latencies: {[0,100]}
Jitter bound: 2

Input: @271 b

Verdict: INCONCLUSIVE
Positive:
Consistent latencies: {[71,75)}
Jitter bound: 2
Negative:
Consistent latencies: {[0,100]}
Jitter bound: 2
\end{lstlisting}

For realistic experiments, we took traces and properties from the gear controller model in \cite{gear}. The model, along with its formal requirements, was created by the company Mecel.

To demonstrate the tightening of bounds for the latency parameters, we monitor the response property 
\[\varphi_{\textit{gear}} = G_{[0, \infty]}(\textit{ReqNewGear} \rightarrow F_{[150, 1205]} \textit{NewGear})\]
that requires the gear controller to change gear within 150 ms to 1205 ms after a shift request. In this experiment the event~\textit{ReqNewGear} is an input from the test harness and \textit{NewGear} is an output from the system under test, i.e., we are concerned with active testing where the inputs arrive after some delay at the system to be monitored and the outputs from the system are only observed after some delay.

To verify~$\varphi_{\textit{gear}}$, we start the monitor with the known bounds on the unknown latency and the jitter. 
We apply the monitor to original traces of the gearbox model, which have their timestamps modified by an assigned input-/output latency (which is unknown to the monitor) plus a random jitter value. The traces have a chance of providing an error, which means that the system changes the gear before 150 ms or after 1205 ms.
Because of the unknown exact delay, the monitor is not guaranteed to catch the error and give a $\bot$ verdict immediately; however, it will tighten the latency bounds and may provide a verdict after further observations.

First we monitor $\varphi_{\textit{gear}}$ with latency parameters $\lat_I \in [10, 50]$ and $\lat_O \in [60, 100]$ and a jitter bound of $10$. The actual input and output latency is $45$ and $65$. The trace provided had an error at observation 16 and 22, after which the monitor gave a $\bot$ verdict.

The changing bounds of $\top$-consistent latencies are illustrated in Fig.~\ref{fig:testgraph1} where the left side shows the upper and lower bounds on the input (red, at the top) and output (yellow, at the bottom) latency and the right side shows the bounds on the latencies added together (blue). When a gear is changed outside the bounds of the property, this is an error at the system, and is marked by a dashed vertical line at the index of that observation. As the first error occurs (at time 16), all  lower bounds on latencies are tightened significantly. The actual latency of the output (which is 65) here already becomes inconsistent with the $\top$ verdict, but as the actual latency is unknown, this remains undetectable and the verdict stays inconclusive at this point of time. Yet another error occurs at the 22nd event, at which point the set of consistent delays for the $\top$ verdict becomes empty and the $\bot$ verdict is given.

\begin{figure}[]
    \centering
        \begin{tikzpicture}
            \begin{axis}[
                legend style={at={(0.5,1.1)}, anchor=south},
                height=5cm,
                width=6cm,
                xlabel={\# Observation},
                ytick={10, 45, 65, 100},
                xtick={1, 16, 22},
                ymajorgrids=true,
                legend entries={
                Bounds on $\lat_O$,
                Bounds on $\lat_I$
                }]
                \addlegendimage{only marks, mark=square*, color=yellow}
                \addlegendimage{only marks, mark=square*, color=red}
                \addplot[draw=none, name path=inlow] table [x=Obs, y=Posinup, col sep=comma] {data-in-out1.csv};
                \addplot[draw=none, name path=inup] table [x=Obs, y=Posinlow, col sep=comma] {data-in-out1.csv};
                \addplot[draw=none, name path=outlow] table [x=Obs, y=Posoutup, col sep=comma] {data-in-out1.csv};
                \addplot[name path=outup, draw=none] table [x=Obs, y=Posoutlow, col sep=comma] {data-in-out1.csv};
                \addplot[red] fill between[
                    of = inup and inlow, 
                    split, 
                  ];
                \addplot[yellow] fill between[
                    of = outup and outlow, 
                    split
                  ];
                \addplot[color=black, very thick, dotted] table[row sep = crcr]{16 0 \\ 16 110 \\};
                \addplot[color=black, very thick, dotted] table[row sep = crcr]{22 0 \\ 22 110 \\};
            \end{axis}
        \end{tikzpicture}
        \begin{tikzpicture}
            \begin{axis}[
                legend style={at={(0.5,1.1)}, anchor=south},
                height=5cm,
                width=6cm,
                xlabel={\# Observations},
                ytick={0, 70, 110, 150},
                xtick={1, 16, 22},
                ymajorgrids=true,
                legend entries={Bounds on $\lat_O + \lat_I$}]
                \addlegendimage{only marks, mark=square*, color=blue}
                \addplot[draw=none, name path=plot1] table [x=Obs, y=Posdifflow, col sep=comma] {data-in-out1.csv};
                \addplot[draw=none, name path=plot2] table [x=Obs, y=Posdiffup, col sep=comma] {data-in-out1.csv};
                \addplot[blue] fill between[
                    of = plot1 and plot2, 
                     split,
                  ];
                \addplot[color=black,very thick, dotted] table[row sep = crcr]{16 0 \\ 16 160 \\};
                \addplot[color=black,very thick, dotted] table[row sep = crcr]{22 0 \\ 22 160 \\};
            \end{axis}
        \end{tikzpicture}
    \caption{Test run with 22 observations ending in a $\bot$ verdict where $\lat_I \in [10, 50], \lat_O \in [60, 100], \jit = 10$ and the actual input latency is $65$ and the actual output latency is $45$. Errors occurred at the 16.\ and 22.\ observation. The maximal size of reach-sets is five states.}
    \label{fig:testgraph1}
\end{figure}

We monitor $\varphi_{\textit{gear}}$ again, now with different latency bounds $\lat_I \in [0, 90]$ and $\lat_O \in [100, 200]$ and the actual input-/output latency is 60 and 120.

Fig.~\ref{fig:testgraph2} shows the changing latency bounds where a single error occurs in the first response (observation 2), after which the upper bounds of the input-/output latency lower slightly, while the upper bound on the combined latency is tightened significantly. It drops slightly below the actual combined latency, which is 180, but this remains undetectable as the actual latency is unknown. Over the course of subsequent observations, the lower bound on the combined latency increases repeatedly, even though no further erroneous behavior is observed. After 22 observations, the set of consistent combined latencies finally becomes empty, and consequently the verdict $\bot$ is given. Note that this conclusive verdict is a consequence of a complex temporal matching process that tries to align the correct observations between time instants 3 and 22 with the possible observation delays and only thereby conclusively detects the single error that occurred at the very beginning of the observation.

\begin{figure}[]
    \centering
        \begin{tikzpicture}
            \begin{axis}[
                legend style={at={(0.5,1.1)}, anchor=south},
                height=5cm,
                width=6cm,
                xlabel={\# Observations},
                ymajorgrids=true,
                legend entries={
                Bounds on $\lat_O$,
                Bounds on $\lat_I$,
                },
                xtick={2, 22},
                ytick={0,60,120,200}
                ]
                \addlegendimage{only marks, mark=square*, color=yellow}
                \addlegendimage{only marks, mark=square*, color=red}
                \addplot[draw=none, name path=inlow] table [x=Obs, y=Posinup, col sep=comma] {data-in-out-diff1.csv};
                \addplot[draw=none, name path=inup] table [x=Obs, y=Posinlow, col sep=comma] {data-in-out-diff1.csv};
                \addplot[draw=none, name path=outlow] table [x=Obs, y=Posoutup, col sep=comma] {data-in-out-diff1.csv};
                \addplot[draw=none, name path=outup] table [x=Obs, y=Posoutlow, col sep=comma] {data-in-out-diff1.csv};
                \addplot[red] fill between[
                    of = inlow and inup, 
                    split
                  ];
                \addplot[yellow] fill between[
                    of = outlow and outup, 
                    split
                  ];
                \addplot[color=black, thick, dotted] table[row sep = crcr]{2 0 \\ 2 200 \\};
            \end{axis}
        \end{tikzpicture}
        \begin{tikzpicture}
            \begin{axis}[
                legend style={at={(0.5,1.1)}, anchor=south},
                height=5cm,
                width=6cm,
                xlabel={\# Observations},
                ytick={0, 100, 180, 290},
                ymajorgrids=true,
                xtick={2, 8, 12, 22},
                legend entries={Bounds on $\lat_O + \lat_I$}]
                \addlegendimage{only marks, mark=square*, color=blue}
                \addplot[draw=none, name path=plot1] table [x=Obs, y=Posdifflow, col sep=comma] {data-in-out-diff1.csv};
                \addplot[draw=none, name path=plot2] table [x=Obs, y=Posdiffup, col sep=comma] {data-in-out-diff1.csv};
                \addplot[blue] fill between[
                    of = plot1 and plot2, 
                    split
                  ];
                \addplot[color=black, thick, dotted] table[row sep = crcr]{2 0 \\ 2 300 \\};
            \end{axis}
        \end{tikzpicture}
    \caption{Test run with 22 observations ending in a $\bot$ verdict where $\lat_I \in [0, 90], \lat_O \in [100, 200], \jit = 10$ and the actual input latency is $60$ and the actual output latency is $120$. A single error occurred at observation \# 2. The maximal size of reach-sets is five states.}
    \label{fig:testgraph2}
\end{figure}

To examine the performance of our monitoring algorithms, we again monitor the property $\varphi_{\textit{gear}}$, but in this case over a trace that does not feature any errors such that monitoring will not terminate early with a conclusive verdict. As the monitor does not terminate, we can evaluate the monitor performance over arbitrarily long traces ranging from 1,000 to 10,000 events.
In these experiments, we compare three scenarios: \textit{classic} is regular monitoring with no delays, \textit{ delay} is where every observation is an output, and \textit{ testing} is where the request is a delayed input and the response is a delayed output. In the delay and testing cases the delay parameters are $\lat_I, \lat_O \in [0, 100]$ and $\jit_I, \jit_O = 10$. 

The maximal response time and number of symbolic states in the reach-set of the three scenarios are plotted in Table~\ref{res:small}.
Response time is the time it takes from observing an event to providing a verdict. The size of the reach-set is the number of symbolic states stored in the representation of the reach-set. In the implementation, the reach-set is stored as a set of symbolic states which might contain redundant information. To minimize the size we check for zone inclusion, such that a symbolic state $(q, Z)$ is only stored if it is not included in another stored state $(q, Z')$ with $Z \subseteq Z'$. We also employ the inactive clock abstraction from \cite{DBLP:conf/rtss/DawsY96}.


The size of the reach-set obviously affects the response time negatively, as is evident from the differences between the classic, delay and testing cases. However, all response times remain below 300 $\mu$s, demonstrating that the implementation can sustain real-time monitoring of the response property over an arbitrary number of observations and in all scenarios, even in the complex testing scenario with both input and output delays.
It is also interesting to note that in the testing case, the maximal number of states increased significantly. This is due to the fact that we add two clocks that are continuously changed but never reset. 
The property itself only requires a single clock, so the total number of clocks in the classic, delay and testing scenarios are 1, 3, and 4.
The stored symbolic states then have more constraints that can be different from each other. If resulting zones have slightly shifted bounds between $etime_I$ and $etime_O$, it might be beneficial to check if any zones can be merged, in order to make the reach-set representation more memory-efficient.

\begin{table}[tb]
  \caption{Results for monitoring $\varphi_{\textit{gear}}$ under delay and test over traces with varying length.}
  \label{res:small}
  \begin{center}
  \renewcommand{\tabcolsep}{0.15cm}
  \begin{tabular}{r r r r r r r}
    \# Observ. & \multicolumn{3}{c}{Max.\ response\ time ($\mu$s)} & \multicolumn{3}{c}{Max. \# Symbolic States}\\
    & Classic & Delay & Testing & Classic & Delay & Testing \\
    \midrule
    1000 & 66 & 82 & 217 & 2 & 3 & 11\\
    2000 & 42 & 76 & 153 & 2 & 3 & 11\\
    3000 & 63 & 83 & 233 & 2 & 3 & 11\\
    4000 & 34 & 84 & 230 & 2 & 3 & 11\\
    5000 & 63 & 68 & 216 & 2 & 3 & 11\\
    6000 & 67 & 101 & 216 &2 &  3 & 11\\
    7000 & 66 & 64 & 228 & 2 & 3 & 11\\
    8000 & 59 & 73 & 215 & 2 & 3 & 11\\
    9000 & 77 & 70 & 256 & 2 & 3 & 11\\
    10000 & 47 & 77 & 216 & 2 & 3 & 11\\
\end{tabular}\end{center}\end{table}

\section{Conclusion}
\label{sec:conc}
We have introduced zone-based algorithms realizing optimal (in the sense of being anticipating~\cite{bauer2006monitoring} as well as conclusive under uncertainty \cite{FFKK22}) online operational monitoring and testing of embedded real-time systems when the communication between the monitor (or testing harness, respectively) and the system is subject to unknown (up to bounds) delay. 
This situation is rather typical in practice as observations are mediated by sensors, may involve conversion between analog and digital, or pass communication networks and consequently are indirect in general, leading to delays and inexact time-stamping. Our constructions thus fill a gap in the pre-existing theories for monitoring and testing hard real-time systems, which tend to assume full and exact temporal observability by immediate coupling or, equivalently, perfect synchrony between systems and their monitors or test harnesses. 

We assume no knowledge of timing information of observations, except knowing when they are received. However, one could imagine a scenario where some output observations include timing information from the system. This information would give us the exact delay (latency + jitter) of a specific observation and, while not providing an exact synthesis, could be used to refine the latency parameter. Thus the parameter would not necessarily be fully determined, and is still needed in the gaps between future observations and for observations without timing information.

A notable point of our construction is that it applies a reduction to simple timed automata and is purely zone-based despite the unknown communication delay being a timing parameter. The construction thus not only avoids the complexities of property analysis for parameterized timed automata \cite{AndreLR22}, but also provides an instance of monitoring and testing under uncertainty where the underlying arithmetic constraint systems remain of fixed dimensionality (namely the number of clocks in the property automata plus two for monitoring) despite their history dependence. This is in stark contrast to direct constraint encodings growing linearly over history length as in \cite{FFKK22}.

In further research, we study the question of monitorability~\cite{bauer2006monitoring}: some properties will never give definitive verdicts (e.g., \myquot{infinitely often $a$}) and are therefore not useful for monitoring. We conjecture that our zone-based approach can be exploited to decide monitorability of real-time properties.

\paragraph*{Acknowledgments}
M.\ Fränzle has been funded by by the State of Lower Saxony, Zukunftslabor Mobilität, and by Deutsche Forschungsgemeinschaft, grants FR\,2715/5-1 and FR 2715/6-1. T.M.\ Grosen and K.G.\ Larsen  have been funded by the VILLUM Investigator grant S4OS, and together with  M.\ Zimmermann they have been supported by DIREC - Digital Research Centre Denmark.

\bibliographystyle{ACM-Reference-Format}
\bibliography{bib}

\end{document}

%% file: packages.tex
\usepackage[T1]{fontenc}

\usepackage{amsmath}

\usepackage{xspace}

\usepackage{xcolor}
\usepackage{tikz}
\usetikzlibrary{math, automata, positioning}
\usepackage{hyperref}
\usepackage{booktabs}
\usepackage{subfig}
\usepackage{listings}

\usepackage{wrapfig}
\usepackage{bbding}
\usepackage{pgfplots}
\usepgfplotslibrary{fillbetween}
\usetikzlibrary{patterns}
\pgfplotsset{compat=1.18}
\usetikzlibrary{plotmarks}

%% file: macros.tex
\newcommand{\nats}{\mathbb{N}}
\newcommand{\reals}{\mathbb{R}}
\newcommand{\nnreals}{\reals_{\geq 0}}
\newcommand{\rats}{\mathbb{Q}}
\newcommand{\nnrats}{\rats_{\geq 0}}

\newcommand{\aut}{\mathcal{A}}
\newcommand{\lang}[1]{L(#1)}
\newcommand{\infset}[1]{\mathrm{Inf}(#1)}
\newcommand{\transition}[1]{\xrightarrow{(\sigma_{#1},\tau_{#1})}}

\newcommand{\nonempty}[1]{S_{#1}^{\textit{ne}}}
\newcommand{\terminal}[3]{\mathcal{R}_{#1}^{#2}(#3)}
\newcommand{\terminalprime}[3]{\widehat{\mathcal{R}}_{#1}^{#2}(#3)}
\newcommand{\compautomaton}{\overline{\mathcal{A}}}
\newcommand{\tbaty}{\mathbb{A}}

\newcommand{\tba}{TBA\xspace}
\newcommand{\mitl}{MITL\xspace}

\newcommand{\true}{\texttt{true}}

\makeatletter
\newcommand{\shortsim}{%
  \settowidth{\@tempdima}{n}
  \resizebox{\@tempdima}{\height}{$\sim$}%
}
\makeatother

\newcommand{\TSigma}{T\Sigma}

\newcommand{\pow}[1]{2^{#1}}
\newcommand\proj[1]{|_{#1}}

\newcommand{\bools}{\mathbb{B}}
\newcommand{\verdicts}{\bools_3}
\newcommand{\unknown}{\textbf{\textit{?}}}

\newcommand{\globaltime}{{\mathit{time}}}
\newcommand{\etime}{{\mathit{etime}}}

\newcommand{\evalfuncsymbol}{\mathcal{V}\xspace}

\newcommand{\evalfuncd}[3]{\evalfuncsymbol_{#1}(#2)(#3)}

\newcommand{\strict}{\mathrm{el}}
\newcommand{\evalfuncsymbolstrict}{\mathcal{V}^\strict\xspace}

\newcommand{\evalfuncdstrict}[3]{\evalfuncsymbolstrict_{#1}(#2)(#3)}

\newcommand{\testevalfuncsymbol}{\widehat{\mathcal{V}}\xspace}

\newcommand{\testevalfuncd}[3]{\testevalfuncsymbol_{#1}(#2)(#3)}

\newcommand{\monitordelayssymbol}{\mathcal{M}_{\delays}\xspace}
\newcommand{\monitordelays}[2]{\monitordelayssymbol(#1)(#2)}
\newcommand{\monitordelayssymbolprime}{\widehat{\mathcal{M}}_{\delays}\xspace}
\newcommand{\monitordelaysprime}[2]{\monitordelayssymbolprime(#1)(#2)}

\newcommand{\chan}{{\it Chan}}
\newcommand{\GT}{{\mathit GT}}
\newcommand{\testGT}{{\widehat{\mathit GT}}}

\newcommand{\testDelta}{\widehat{\Delta}}
\newcommand{\delays}{\mathcal{D}}
\newcommand{\lat}{\delta}
\newcommand{\jit}{\varepsilon}

\newcommand{\myquot}[1]{``#1''}
\newcommand{\decorate}[1]{#1'}

\newcommand{\post}{\mathrm{Post}}
\newcommand{\succc}{\mathrm{Succ}}

\newcommand\dir{\text{\it direction}}

\newcommand{\IO}{\mathrm{IO}}

%% file: figures/automaton-example.tex
\begin{tikzpicture} [node distance = 1.7cm, thick]
    \node (q0)     [state, initial text={}]          {$q_0$};
    \node (i)   at (0,.85) {};
    \node (q1)     [state, right = of q0]    {$q_1$};
    \node (phi)    [state, right = of q1]    {$\varphi$};
    \node (notphi) [state, left = of q0]    {$\neg\varphi$};
    
    \draw[->, thick] (i) edge (q0);
     \draw[->, thick] (q0) edge node[above] {$a$} node[below]{$x \leq 10$} (q1);
     \draw[->, thick] (q1) edge node[above] {$a, b$} node[below]{$x > 20$} (phi);

     \draw[->, thick] (q0) edge[bend right=30] node[above]{$b$} (notphi);
     \draw[->, thick] (q0) edge[] node[above]{$a$} node[below]{$x > 10$} (notphi);

     \draw[->, thick] (q1) edge[bend left=30] node[above]{$b$}  node[below]{$x\leq20$} (notphi);
    
    \draw[->, thick] (phi) edge [loop below] node[] {$a,b$} (phi);
    \draw[->, thick] (q1) edge [loop below] node[above right=0.1cm] {$a$} node[right] {$x \leq 20$} (q1);
    \draw[->, thick] (notphi) edge[loop below] node[]{$a,b$} (notphi);

\end{tikzpicture}

%% file: figures/zone-based-one-way-jitter.tex
\resizebox{\linewidth}{!} {
    \begin{tikzpicture}
        \tikzmath{\y=5.5; \x=4; \l=0.3; \u=1.5; \enabledlow=1; \enabledup=3;}
        \filldraw[black!20] (\enabledlow,0) rectangle (\enabledup,\y);

        \draw[very thick, ->] (0,0) -- (0,\y) node[below right]{etime};
        \draw[very thick, ->] (0,0) -- (\x,0) node[below left]{time};
        
        \draw[ultra thick, black!20!green] (0,\l) -- (0,\u);
        \draw (0,\u) node[right]{$u$};
        \draw (0,\l) node[right]{$l$};
        
        \draw[dashed] (0,\u) -- (\x,\x+\u);
        \draw[dashed] (0,\l) -- (\x,\x+\l);
        
        \draw (2,0.25) node[]{$\sigma$ enabled};
        \draw[dashed, latex reversed-latex reversed] (\enabledlow,0.5) -- (\enabledup,0.5);
        
        \draw[thick, blue, fill=blue!70!green!20] (\enabledlow, \enabledlow+\l) -- (\enabledup, \enabledup+\l) -- (\enabledup, \enabledup+\u) -- (\enabledlow, \enabledlow+\u) -- cycle;
        
        \draw[very thin] (2, 2+\u) -- (1.7, 4.3);
        \draw[very thin] (2, 2+\l) -- (2.3, 1.4);

        \draw (1.5,4.5) node[]{$etime - time \leq u$};
        \draw (2.5,1.2) node[]{$time - etime \leq -l$};
        
    \end{tikzpicture}





\hspace{10pt}
    \begin{tikzpicture}
        \tikzmath{\y=5.5; \x=4; \l=0.3; \u=1.5; \enabledlow=1; \enabledup=3; \t=3.9; \e=0.2;}

        \draw[very thick, ->] (0,0) -- (0,\y) node[below right]{etime};
        \draw[very thick, ->] (0,0) -- (\x,0) node[below left]{time};

        \draw[thick, blue, fill=blue!70!green!20] (\enabledlow, \enabledlow+\l) -- (\enabledup, \enabledup+\l) -- (\enabledup, \enabledup+\u) -- (\enabledlow, \enabledlow+\u) -- cycle;
        
        
        \draw[dashed] (0,\t+\e) node[left]{$\tau$} node[above right=0.5]{$etime \in [\tau-\jit, \tau]$} -- (\x, \t+\e);
        \draw[dashed] (0,\t-\e) node[left]{$\tau-\jit$} -- (\x, \t-\e);
        \draw[very thick, red, fill=red!50] (\t-\u-\e,\t-\e) -- (\enabledup,\t-\e) -- (\enabledup, \t+\e) -- (\t+\e-\u, \t+\e) -- cycle;

        \draw[dashed] (\t-\u-\e,\t-\e) -- (\t-\u-\e, 0.8);
        \draw[dashed] (\enabledup,\t-\e) -- (\enabledup, 0.8);

        \draw[dashed] (\t-\u-\e,0.5) -- (\t-\u-\e, 0);
        \draw[dashed] (\enabledup,0.5) -- (\enabledup, 0);

        \draw[dashed, latex reversed-latex reversed] (\t-\u-\e,1) -- (\enabledup,1);
        \draw (\t-\u-\e+0.4, 0.7) node[]{Event at sys-time};

    \end{tikzpicture}
    \hspace{10pt}
    \begin{tikzpicture}
        \tikzmath{\y=5.5; \x=4; \l=0.3; \u=1.5; \enabledlow=1; \enabledup=3; \t=3.9; \e=0.2;}

        \draw[very thick, ->] (0,0) -- (0,\y) node[below right]{etime};
        \draw[very thick, ->] (0,0) -- (\x,0) node[below left]{time};
        
        \draw[dashed] (\enabledlow, \enabledlow+\l) -- (\enabledup, \enabledup+\l) -- (\enabledup, \enabledup+\u) -- (\enabledlow, \enabledlow+\u) -- cycle;

        \draw[dashed] (0,\t-\e) node[left]{$\tau-\jit$} -- (\x, \t-\e);

        \draw[thick, black!60!green, fill=green!20] (\t-\u-\e,\t-\e) -- (\enabledup,\t-\e) -- (\x,\x-\enabledup+\t-\e) -- (\x, \x+\u) -- cycle;

    \end{tikzpicture}
}